\documentclass{IEEEtaes}

\usepackage[colorlinks,urlcolor=black,linkcolor=black,citecolor=black]{hyperref}

\usepackage{color,array,amsthm}

\usepackage{graphicx}

\usepackage{cite, url}
\usepackage{balance}
\usepackage{bm}
\usepackage{amsmath,amssymb,amsfonts}
\usepackage{subfigure}

\usepackage{algorithm}
\usepackage[noend]{algpseudocode}

\jvol{XX}
\jnum{XX}
\jmonth{XXXXX}
\paper{1234567}
\pubyear{2022}
\doiinfo{TAES.2020.Doi Number}

\newtheorem{theorem}{Theorem}
\newtheorem{lemma}{Lemma}

\newtheorem{remark}{Remark}


\begin{document}

\title{Low-Rank and Row-Sparse Decomposition for Joint DOA Estimation and Distorted Sensor Detection}

%
%
%

\author{Huiping Huang}
\member{Student Member, IEEE}
\affil{Darmstadt University of Technology, Germany} 

\author{Qi Liu}
\member{Member, IEEE}
\affil{South China University of Technology, China; \\
Pazhou Lab, Guangzhou 510330, China} 

\author{Hing C. So}
\member{Fellow, IEEE}
\affil{City University of Hong Kong, China}

 \author{Abdelhak M. Zoubir}
\member{Fellow, IEEE}
 \affil{Darmstadt University of Technology, Germany}

\receiveddate{This manuscript is submitted for review on XX.}

\corresp{ {\itshape (Corresponding author: H. Huang)}. }

\authoraddress{H. Huang and A. M. Zoubir are with Department of Electrical Engineering and Information Technology, Darmstadt University of Technology, Germany (emails: h.huang@spg.tu-darmstadt.de, zoubir@spg.tu-darmstadt.de). Q. Liu is with School of Future Technology, South China University of Technology, China, and also with Pazhou Lab, Guangzhou 510330, China (email: drliuqi@scut.edu.cn). H. C. So is with Department of Electrical Engineering, City University of Hong Kong, China (email: hcso@ee.cityu.edu.hk).}


\markboth{H. Huang \textit{et al.}}{LR$^{2}$SD FOR DOA ESTIMATION DISTORTED SENSOR DETECTION}
\maketitle

\begin{abstract}Distorted sensors could occur randomly and may lead to the breakdown of a sensor array system. We consider an array model within which a small number of sensors are distorted by unknown sensor gain and phase errors. With such an array model, the problem of joint direction-of-arrival (DOA) estimation and distorted sensor detection is formulated under the framework of low-rank and row-sparse decomposition. We derive an iteratively reweighted least squares (IRLS) algorithm to solve the resulting problem in both noiseless and noisy cases. The convergence property of the IRLS algorithm is analyzed by means of the monotonicity and boundedness of the objective function. Extensive simulations are conducted regarding parameter selection, convergence speed, computational complexity, and performances of DOA estimation as well as distorted sensor detection. Even though the IRLS algorithm is slightly worse than the alternating direction method of multipliers in detecting the distorted sensors, the results show that our approach outperforms several state-of-the-art techniques in terms of convergence speed, computational cost, and DOA estimation performance.
\end{abstract}


\section{INTRODUCTION}
\label{introduction}

D{\scshape irection}-of-arrival (DOA) estimation is one of the most important topics in array signal processing, which has found numerous applications in radar, sonar, wireless communications, to name just a few \cite{Krim1996, VanTrees2002, Viberg2014}. Many classical approaches have been proposed, including multiple signal classification (MUSIC) \cite{Schmidt1986}, estimation of signal parameters via rotational invariance techniques (ESPRIT) \cite{Roy1989}, and maximum likelihood methods \cite{Bohme1986, Ziskind1988}. However, it is known that most of these high-resolution algorithms rely heavily on the exact knowledge of the array manifold, and hence their performance may greatly suffer when the sensor array encounters distortions \cite{Vorobyov2003, Wang2017Jul, Wang2017Dec, Yang2017, Ramamohan2018, Huang2019}, such as unknown sensor gain and phase uncertainties, which is the focus of this work. {\color{black}More recently, techniques based on low-rank and sparse matrix decomposition have been applied to DOA estimation or tracking, see e.g. \cite{Lin2015, Das2017, Markopoulos2014, Markopoulos2019}. However, these works merely consider the well-calibrated array, and they are not straightforwardly applicable to an array with sensor errors.}

There is a large number of works devoted to handle distorted or completely failed sensors \cite{Yeo1999, Vigneshwaran2007, Muma2012, Oliveri2012, Zhu2015, Wang2017, Liu2019June, Stoica1996, Ng2009, Jiang2013, Pesavento2002, See2004, Liao2012, Steffens2014, Liao2014, Suleiman2018, Afkhaminia2017, Lee2019, Huang2021}. In \cite{Yeo1999}, the genetic algorithm \cite{Holland1992} was applied for array failure correction. A minimal resource allocation network was used for DOA estimation under array sensor failure \cite{Vigneshwaran2007}, which requires a training procedure with no failed sensors. A Bayesian compressive sensing approach was proposed in \cite{Oliveri2012}, which needs a noise-free array as a reference. Methods using difference co-array were developed in \cite{Zhu2015, Wang2017, Liu2019June}. The idea of \cite{Zhu2015} was based on the fact that positions corresponding to damaged sensors may be occupied by virtual sensors and thus the impact of sensor failure could be avoided. However, this is not applicable when the failed sensors are located on the first or last position of the array, or when the malfunctioned sensors occur on symmetrical positions of the array, in which situations there exist \textit{holes} in the difference co-array. On the other hand, \cite{Wang2017} and \cite{Liu2019June} restricted the array to some special sparse structures, such as co-prime and nested arrays. Approaches based on pre-calibrated sensors have been well-documented in the past decades \cite{Pesavento2002, See2004, Liao2012, Steffens2014, Liao2014, Suleiman2018}. These methods require the knowledge of the calibrated sensors and they are time- and energy-consuming. 


To circumvent the above-mentioned shortcomings, and to tackle the DOA estimation problem with an array in which a few sensors are distorted by unknown sensor gain and phase uncertainties, we formulate the problem under the framework of low-rank and row-sparse decomposition (LR$^{2}$SD), which can be regarded as a special structure of low-rank and sparse decomposition (LRSD). {\color{black}Note that LRSD is also known as robust principal component analysis (RPCA) \cite{Vaswani2018Jul, Vaswani2018Aug, Bouwmans2018}.} The LRSD technique has become a popular tool in finding a low-dimensional subspace from sparsely and arbitrarily corrupted observations, and it has wide applications in science and engineering, ranging from bioinformatics, web search, to imaging, audio and video processing \cite{Jolliffe1986_chapter, Wright2009, Zhang2011, Lin2013, Bando2018}. Another special structure of LRSD is low-rank and column-sparse decomposition (LRCSD) \cite{Xu2012, Liu2010, Liu2013, Lu2015(b), Liu2019}, {\color{black}also known as RPCA-outlier pursuit \cite{Zhang2015, Li2018, Guyon2012, Rodreguez2014},} which has been recently proposed to handle the scenarios where corruptions take place column-sparsely, meaning that the corruption matrix is column-wise sparse. Such situations occur for example when a fraction of the data vectors are grossly corrupted by outliers \cite{Liu2013, Liu2019}.

Several algorithms have been contributed to solve the LRSD and LRCSD problems, such as singular value thresholding (SVT) \cite{Cai2010}, accelerated proximal gradient (APG) \cite{Beck2009}, alternating direction method of multipliers (ADMM) \cite{Lin2013, Liu2019}, and iteratively reweighted least squares (IRLS) \cite{Lu2015(b), Liu2019, Guyon2012, Rodreguez2014}. The SVT, APG, and ADMM methods will be reviewed in Section \ref{relatedmethods} in the context of joint DOA estimation and distorted sensor detection. The above three methods require one singular value decomposition (SVD) in each iteration, which may be unbearable for large scale problems. Instead, IRLS relies on simple linear algebra, and it generally has a linear convergence rate \cite{Daubechies2009, Ba2014, Ene2019, Straszak2021, Kummerle2021}. In this sense, the IRLS is more efficient in solving the corresponding problems.

Therefore, in the present work, we develop an IRLS algorithm for joint DOA estimation and distorted sensor detection. The main contributions include:
\begin{itemize}
\item Both noiseless and noisy cases are considered. The convergence property of the algorithm is analyzed, via the monotonicity and boundedness of the objective function.
\item The computational complexities of the IRLS algorithm as well as the SVT, APG, and ADMM methods are theoretically analyzed.
\item Extensive simulations are conducted in view of parameter selection, convergence speed, computational time, and performance of DOA estimation and distorted sensor detection.
\end{itemize}

The remainder of the paper is organized as follows. The signal model and problem statement are established in Section \ref{signalmodel}. A review of state-of-the-art works is provided in Section \ref{relatedmethods}. Section \ref{proposedmethod} derives an IRLS algorithm for joint DOA estimation and distorted sensor detection. Numerical results are given in Section \ref{simulation}, while Section \ref{conclusion} concludes this paper.

\textit{Notation:} In this paper, bold-faced lower-case and upper-case letters stand for vectors and matrices, respectively. Superscripts $\cdot^{\rm{T}}$ and $\cdot^{\rm{H}}$ denote transpose and Hermitian transpose, respectively. $\mathbb{C}$ is the set of complex numbers, and $\jmath = \sqrt{-1}$. For a real-valued scalar $a$, $|a|$ denotes its absolute value. The minimum value of two scalars $a$ and $b$ is denoted as $\min\{a, b\}$. $\|\cdot\|_{2}$ is the $\ell_{2}$ norm of a vector. $\|\cdot\|_{\rm{F}}$ and $\|\cdot\|_{*}$ represent the Frobenius norm and the nuclear norm (sum of singular values) of a matrix, respectively. $\|\cdot\|_{2,0}$ and $\|\cdot\|_{2,1}$ denote the $\ell_{2,0}$ mixed-norm and $\ell_{2,1}$ mixed-norm of a matrix, respectively, whose definitions are given as $\|{\bf V}\|_{2,0} \triangleq \text{card}(\{\|{\bf V}_{i,:}\|_{2}\})$ and $\|{\bf V}\|_{2,1} \triangleq \sum_{i = 1}^{M}\|{\bf V}_{i,:}\|_{2}$, for ${\bf V} \in \mathbb{C}^{M \times T}$, where $\text{card}(\cdot)$ is the cardinality of a set, $\{\|{\bf V}_{i,:}\|_{2}\} = \{ \|{\bf V}_{1,:}\|_{2}, \|{\bf V}_{2,:}\|_{2}, \cdots, \|{\bf V}_{M,:}\|_{2} \}$, and ${\bf V}_{i,:}$ is the $i$th row of ${\bf V}$. $\rm{rank}(\cdot)$ is the rank operator, defined as $ \text{rank}({\bf Z}) \triangleq \text{card}(\{\sigma_{i}({\bf Z})\})$, with $\sigma_{i}({\bf Z})$ being the $i$th singular value of ${\bf Z}$ and $\{\sigma_{i}({\bf Z})\}$ denoting the set containing all singular values of ${\bf Z}$. For two matrices ${\bf X}$ and ${\bf Y}$ of the same dimensions, we define their Frobenius inner product as $\langle {\bf X} , {\bf Y} \rangle \triangleq \text{trace}({\bf X}^{\text{H}}{\bf Y})$, where $\rm{trace}(\cdot)$ denotes the trace of a square matrix.

\section{SIGNAL MODEL AND PROBLEM STATEMENT}
\label{signalmodel}

Suppose that a linear antenna array of $M$ sensors receives $K$ far-field narrowband signals from directions ${\bm \theta} = [\theta_{1}, \theta_{2}, \cdots, \theta_{K}]^{\text{T}}$. The antenna array of interest is assumed to be randomly and \textit{sparsely} distorted by sensor gain and phase uncertainty (the number of distorted sensors is far smaller than $M$). Further, we assume that the number of distorted sensors and their positions are unknown. Fig. \ref{SCA} illustrates the array model, where the black circles stand for \textit{perfect} sensors and the boxes refer to distorted ones. The boxes appear randomly and sparsely within the whole linear array.

\begin{figure}[t]
\centering
{\includegraphics[width=0.5\textwidth]{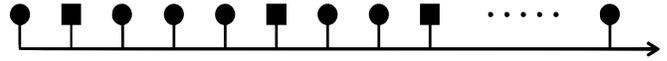}} \vspace{-4mm}
\caption{Illustration of array structure of interest.}
\label{SCA}
\end{figure}

The array observation can be written as
\begin{align*}
{\bf y}(t) = \breve{\bf \Gamma}{\bf A}{\bf s}(t) + {\bf n}(t) \triangleq ({\bf I} + {\bf \Gamma}){\bf A}{\bf s}(t) + {\bf n}(t),
\end{align*}
where $t = 1, 2, \cdots, T$ denotes the time index, $T$ is the total number of available snapshots, ${\bf s}(t) \in \mathbb{C}^{K}$ and ${\bf n}(t) \in \mathbb{C}^{M}$ are signal and noise vectors, respectively. The steering matrix ${\bf A} = [{\bf a}(\theta_{1}), {\bf a}(\theta_{2}), \cdots, {\bf a}(\theta_{K})] \in \mathbb{C}^{M \times K}$ has steering vectors as columns, where the steering vector ${\bf a}(\theta_{k})$ is a function of $\theta_{k}$, for $k = 1, 2, \cdots, K$. In addition, $\breve{\bf \Gamma} \triangleq {\bf I} + {\bf \Gamma}$ indicates the electronic sensor status (either perfect or distorted), where ${\bf I}$ is the $M \times M$ identity matrix, and ${\bm \Gamma}$ is a diagonal matrix with its main diagonal, ${\bm \gamma} = [\gamma_{1}, \gamma_{2}, \cdots, \gamma_{M}]^{\text{T}}$, being a sparse vector. Specifically, for $m = 1, 2, \cdots, M$
\begin{align*}
\gamma_{m} \left\{ \begin{array}{l}
			\!\! = 0, ~ \text{if the $m$th sensor is perfect,} \\
			\!\! \neq 0, ~ \text{if the $m$th sensor is distorted.}
\end{array} \right.
\end{align*}
The non-zero $\gamma_{m}$ denotes sensor gain and phase error, namely, $\gamma_{m} = \rho_{m}e^{\jmath \phi_{m}}$, where $\rho_{m}$ and $\phi_{m}$ are the gain and phase errors of the $m$th sensor, respectively.

Collecting all the snapshots into a matrix, we have
\begin{align}
\label{datamodel_matrix}
{\bf Y} = ({\bf I} + {\bm \Gamma}){\bf A}{\bf S} + {\bf N},
\end{align}
where ${\bf Y} = [{\bf y}(1), {\bf y}(2), \cdots, {\bf y}(T)] \in \mathbb{C}^{M \times T}$ contains the {\color{black}measurements}, ${\bf S} = [{\bf s}(1), {\bf s}(2), \cdots, {\bf s}(T)] \in \mathbb{C}^{K \times T}$ denotes the signal matrix, and ${\bf N} = [{\bf n}(1), {\bf n}(2), \cdots, {\bf n}(T)] \in \mathbb{C}^{M \times T}$ is the noise matrix. Defining ${\bf Z} \triangleq {\bf A}{\bf S}$ and ${\bf V} \triangleq {\bm \Gamma}{\bf A}{\bf S}$, (\ref{datamodel_matrix}) becomes:
\begin{align}
\label{datamodel_ZV}
{\bf Y} = {\bf Z} + {\bf V} + {\bf N},
\end{align}
where ${\bf Z} \in \mathbb{C}^{M \times T}$ is a low-rank matrix of rank $K$ (in general $K < \min\{M, T\}$), and ${\bf V} \in \mathbb{C}^{M \times T}$ is a row-sparse (meaning that only a few rows are non-zero) matrix due to the sparsity of the main diagonal of ${\bm \Gamma}$.

Given the array {\color{black}measurements} ${\bf Y}$, our task is to simultaneously estimate the incoming directions of signals and detect the distorted sensors within the array. Note that the number of distorted sensors is small, but unknown, and their positions are unknown as well.


{\color{black}
\section{RELATED WORKS}
\label{relatedmethods}
Related works for solving the joint DOA estimation and distorted sensor detection include SVT, APG, and ADMM. The SVT method was first proposed for matrix completion, see for example \cite{Cai2010}. By adapting the SVT algorithm to our problem, we need to solve
\begin{align}
\min_{{\bf Z}, {\bf V}, {\bf W}} ~ \|{\bf Z}\|_{*} & + \lambda\|{\bf V}\|_{2,1} + \frac{1}{2\tau}\|{\bf Z}\|_{\text{F}}^{2} \nonumber  \\
& + \frac{1}{2\tau}\|{\bf V}\|_{\text{F}}^{2} + \frac{1}{\tau}\langle {\bf W} , {\bf Y} \! - \! {\bf Z} \! - \! {\bf V} \rangle,
\end{align}
where $\lambda$ is a tuning parameter, $\tau$ is a large positive scalar such that the objective function is perturbed slightly. The SVT approach iteratively updates ${\bf Z}$, ${\bf V}$, and ${\bf W}$. ${\bf Z}$ and ${\bf V}$ are updated by solving the above problem with ${\bf W}$ fixed. Then ${\bf W}$ is updated as ${\bf W} = {\bf Y} - {\bf Z} - {\bf V}$. The following well-known results are used when updating ${\bf Z}$ and ${\bf V}$ \cite{Cai2010}:
\begin{align*}
{\bf L}\mathcal{S}_{\kappa}({\bf S}){\bf R}^{\mathrm{H}} = & \arg\min_{{\bf X}} ~ \kappa\|{\bf X}\|_{*} + \frac{1}{2}\|{\bf X} - {\bf C}\|_{\mathrm{F}}^{2}, \\
\mathcal{S}_{\kappa}({\bf C}) = & \arg\min_{{\bf X}} ~ \kappa\|{\bf X}\|_{2,1} + \frac{1}{2}\|{\bf X} - {\bf C}\|_{\mathrm{F}}^{2},
\end{align*}
where ${\bf L}{\bf S}{\bf R}^{\mathrm{H}}$ is the SVD of ${\bf C}$ and the element-wise soft-thresholding operator is defined as:
\begin{align*}
\mathcal{S}_{\kappa}(x) = \left\{ 
\begin{array}{l}
x - \kappa, \text{~if~} x > \kappa, \\
x + \kappa, \text{~if~} x < \kappa, \\
0, \text{~otherwise},
\end{array}
\right.
\end{align*}
with parameter $\kappa > 0$. The applicability of SVT is limited since it is difficult to select the step size for speedup \cite{Lin2013}.

The second method is APG, whose updating equation can be given as \cite{Beck2009}
\begin{align}
({\bf Z}_{k+1} , {\bf V}_{k+1}) = \arg\min_{{\bf Z}, {\bf V}} ~ h({\bf Z}, {\bf V})
\end{align}
where subscript $\cdot_{k}$ denotes the variable in the $k$th iteration, $h({\bf Z}, {\bf V}) \triangleq p({\bf Z}_{k}, {\bf V}_{k}) + \langle \nabla_{{\bf Z}_{k}}p({\bf Z}, {\bf V}_{k}) , {\bf Z} - {\bf Z}_{k} \rangle + \langle \nabla_{{\bf V}_{k}}p({\bf Z}_{k}, {\bf V}) , {\bf V} - {\bf V}_{k} \rangle + \mu M \|{\bf Z} + {\bf V} - {\bf Z}_{k} - {\bf V}_{k}\|_{\mathrm{F}}^{2} + q({\bf Z}, {\bf V})$, with $p({\bf Z}, {\bf V}) \triangleq \frac{1}{\mu}\|{\bf Y} - {\bf Z} - {\bf V}\|_{\mathrm{F}}^{2}$, $q({\bf Z}, {\bf V}) \triangleq \|{\bf Z}\|_{*} + \lambda\|{\bf V}\|_{2,1}$, and $\mu$ being a small positive scalar. The detailed algorithm can be found in \cite{Beck2009} and also \cite{Lin2013}.

As for ADMM, we consider the following problem
\begin{align}
\min_{{\bf Z}, {\bf V}} ~ \|{\bf Z}\|_{*} & + \lambda\|{\bf V}\|_{2,1} \quad \text{s.t.} ~ {\bf Y} = {\bf Z} + {\bf V},
\end{align}
and its augmented Lagrangian function is $\mathcal{L}_{\mu}({\bf Z}, {\bf V}, {\bf W}) = \|{\bf Z}\|_{*} + \lambda\|{\bf V}\|_{2,1} + \langle {\bf W} , {\bf Y} \! - \! {\bf Z} \! - \! {\bf V} \rangle + \frac{\mu}{2}\|{\bf Y} \! - \! {\bf Z} \! - \! {\bf V}\|_{\text{F}}^{2}$, where ${\bf W}$ denotes the dual variable and $\mu$ is the augmented Lagrangian parameter. Then ADMM updates ${\bf Z}$, ${\bf V}$, and ${\bf W}$, in a sequential manner. ${\bf Z}$ and ${\bf V}$ are solved by minimizing $\mathcal{L}_{\mu}({\bf Z}, {\bf V}, {\bf W})$ with respect to ${\bf Z}$ (resp. ${\bf V}$) while keeping ${\bf V}$ (resp. {\bf Z}) and ${\bf W}$ unchanged; ${\bf W}$ is updated as ${\bf W} = {\bf W} + \mu({\bf Y} - {\bf Z} - {\bf V})$ \cite{Boyd2011}.

All the aforementioned three algorithms require performing one SVD per iteration. Therefore, their computational complexity is extremely high, especially when the problem size is large. Their convergence speed and computational cost will be compared in Sections \ref{convergencespeed} and \ref{complexity}, respectively.
}

\section{PROPOSED METHOD}
\label{proposedmethod}

In this section, we develop an IRLS algorithm for the task of jointly estimating DOAs of sources and detecting distorted sensors. We start by considering the noiseless case, and then focus on the noisy case.

\subsection{NOISELESS CASE}
\label{noiselesscase}
In the noiseless case, the data model (\ref{datamodel_ZV}) is simplified as ${\bf Y} = {\bf Z} + {\bf V}$.
Therefore, we formulate the following LR$^{2}$SD problem, as
\begin{align}
\min_{{\bf Z}, {\bf V}} ~ \text{rank}({\bf Z}) + \lambda\|{\bf V}\|_{2,0} \quad \text{s.t.} ~ {\bf Y} = {\bf Z} + {\bf V},
\end{align}
where $\lambda$ is a tuning parameter. By substituting the equality constraint into the objective, and replacing the rank and $\ell_{2,0}$ mixed-norm with the nuclear norm and $\ell_{2,1}$ mixed-norm, respectively, we have its convex counterpart, as
\begin{align}
\min_{{\bf Z}} ~ \|{\bf Z}\|_{*} + \lambda\|{\bf Y} \! - \! {\bf Z}\|_{2,1}.
\end{align}

The nuclear norm and the $\ell_{2,1}$ mixed-norm are non-smooth, and thus they are not differentiable at some points. 
%
{\color{black}To deal with this issue, we  introduce a smoothing parameter $\mu$}, and obtain the gradients as 
\begin{align*}
\frac{\partial \|[{\bf Z} , \mu{\bf I}]\|_{*}}{\partial {\bf Z}} & = {\bf P}{\bf Z} \\
\frac{\partial \|[{\bf Y} \! - \! {\bf Z}, \mu{\bf 1}]\|_{2,1}}{\partial {\bf Z}} & =  {\bf Q}({\bf Z} \! - \! {\bf Y}),
\end{align*}
where ${\bf 1}$ is an all-ones vector of appropriate length, ${\bf P} \triangleq \left({\bf Z}{\bf Z}^{\text{H}}  +  \mu^{2}{\bf I}\right)^{-\frac{1}{2}}$ and 
\begin{align}
\label{Q_noiseless}
{\bf Q} \triangleq  \left[ \!\!
\begin{array}{ccc}
\frac{1}{\sqrt{\|({\bf Y} - {\bf Z})_{1,:}\|_{2}^{2} ~\! + ~\! \mu^{2}}} & & \\
 & \!\!\!\! \ddots \!\!\!\! & \\
 & & \frac{1}{\sqrt{\|({\bf Y} - {\bf Z})_{M,:}\|_{2}^{2} ~\! + ~\! \mu^{2}}} 
\end{array} \!\! \right] \! .
\end{align}
The problem to be solved now turns to be
\begin{align}
\label{prob_noiseless}
\min_{{\bf Z}} ~ f({\bf Z}) \triangleq \|[{\bf Z}, \mu{\bf I}]\|_{*} + \lambda\|[{\bf Y} \! - \! {\bf Z}, \mu{\bf 1}]\|_{2,1},
\end{align}
where the objective function $f({\bf Z})$ is differentiable everywhere w.r.t. ${\bf Z}$, as long as $\mu \neq 0$. The derivative of $f({\bf Z})$ w.r.t. ${\bf Z}$ is
\begin{align*}
\frac{\partial f({\bf Z})}{\partial {\bf Z}} = {\bf P}{\bf Z} + \lambda{\bf Q}({\bf Z} \! - \! {\bf Y}).
\end{align*}
According to the Karush–Kuhn–Tucker (KKT) condition, we have ${\bf P}{\bf Z} + \lambda{\bf Q}({\bf Z} \! - \! {\bf Y}) = {\bf 0}$, indicating that ${\bf Z} = \lambda({\bf P} + \lambda{\bf Q})^{-1}{\bf Q}{\bf Y}$. This leads to the IRLS iterative process as
\begin{align}
\label{irls_progress_noiseless}
{\bf Z}_{k+1} = \lambda({\bf P}_{k} + \lambda{\bf Q}_{k})^{-1}{\bf Q}_{k}{\bf Y},
\end{align}
where both ${\bf P}_{k}$ and ${\bf Q}_{k}$ are dependent on ${\bf Z}_{k}$. The IRLS algorithm for the noiseless case is summarized in Algorithm \ref{IRLS_noiseless}, where $\epsilon$ is a small scalar and $k_{\text{max}}$ is a large scalar, used to terminate the algorithm.


\begin{algorithm}[t]
	\caption{IRLS algorithm for noiseless case}
	\label{IRLS_noiseless}
	\hspace*{\algorithmicindent} \textbf{Input~~~~\!:} ${\bf Y} \in \mathbb{C}^{M \times T}$, $\lambda$, $\mu$, $\epsilon$, $k_{\text{max}}$ \\
	\hspace*{\algorithmicindent} \textbf{Output~~\!:} ${\bf{\widehat Z}} \in \mathbb{C}^{M \times T}$, ${\bf{\widehat V}} \in \mathbb{C}^{M \times T}$ \\
	\hspace*{\algorithmicindent} \textbf{Initialize:} ${\bf Z}_{0} \gets {\bf Z}_{\text{init}}$, ${\bf V}_{0} \gets {\bf V}_{\text{init}}$, $k \gets 0$
	\begin{algorithmic}[1]
		\While {not converged}
		\State $k \gets k+1$
		\State calculate ${\bf P}_{k}$ and ${\bf Q}_{k}$ 
		\State update ${\bf Z}_{k}$ using ${\bf Z} = \lambda({\bf P} + \lambda{\bf Q})^{-1}{\bf Q}{\bf Y}$ 
		\State converged $\gets$ $ k \geq k_{\text{max}} $ or {\color{black}$ \frac{| f({\bf Z}_{k}) - f({\bf Z}_{k-1}) |}{| f({\bf Z}_{k}) | } \leq \epsilon $}
		\EndWhile \State \textbf{end while} \\
		${\bf{\widehat Z}} \gets {\bf Z}_{k}$, ${\bf{\widehat V}} \gets {\bf Y} \! - \! {\bf Z}_{k}$
	\end{algorithmic}
\end{algorithm}

\subsection{CONVERGENCE ANALYSIS FOR ALGORITHM \ref{IRLS_noiseless}}
\label{convergence_noiseless}
We first provide two lemmata giving two important inequalities regarding the trace function and the $\ell_{2,1}$ mixed-norm. Then, we prove the monotonicity and the boundedness of the objective function in Problem (\ref{prob_noiseless}).

\begin{lemma}[Lemma 2 in \cite{Lu2015(b)}]
\label{lemmaTrace}
For any two symmetric positive definite matrices ${\bf X}$ and ${\bf Y}$, it holds that ${\rm{trace}} \! \left({\bf Y}^{\frac{1}{2}}\right) - {\rm{trace}} \! \left({\bf X}^{\frac{1}{2}}\right) \geq {\rm{trace}} \! \left(\frac{1}{2}({\bf Y} - {\bf X})^{\rm{H}}{\bf Y}^{-\frac{1}{2}}\right)$.
\end{lemma}

\begin{lemma}
\label{lemmaL21norm}
For any matrices ${\bf X}$ and ${\bf Y}$ $\in \mathbb{C}^{M \times T}$, we have $\|{\bf Y}\|_{2,1} - \|{\bf X}\|_{2,1} \geq \frac{1}{2}\rm{trace} \! \left({\bf H}\left({\bf Y}{\bf Y}^{\rm{H}} - {\bf X}{\bf X}^{\rm{H}}\right)\right)$, where
\begin{align}
\label{eq_lemma}
{\bf H} = \left[ \!\! \begin{array}{ccc}
\frac{1}{\|{\bf Y}_{1,:}\|_{2}} & & \\
 & \!\! \ddots \!\! & \\
 & & \frac{1}{\|{\bf Y}_{M,:}\|_{2}}
\end{array} \!\! \right] \! .
\end{align}
\end{lemma} 

\begin{proof}
Due to the concavity of function $\sqrt{x}$ ($x \geq 0$), we have $\sqrt{y} - \sqrt{x} \geq \frac{1}{2\sqrt{y}}(y - x)$ for all $x \geq 0$ and $y \geq 0$. Therefore, \vspace{-4mm}
\begin{align*}
\|{\bf Y}\|_{2,1} - \|{\bf X}\|_{2,1} & = \sum_{i}^{M} \left[ \sqrt{\|{\bf Y}_{i,:}\|_{2}^{2}} ~ - ~\! \sqrt{\|{\bf X}_{i,:}\|_{2}^{2}} \right] \\
& \geq \sum_{i}^{M} \left[ \frac{1}{2\|{\bf Y}_{i,:}\|_{2}}\left(\|{\bf Y}_{i,:}\|_{2}^{2} - \|{\bf X}_{i,:}\|_{2}^{2}\right) \right] \\
& = \frac{1}{2}\text{trace} \! \left({\bf H}\left({\bf Y}{\bf Y}^{\text{H}} - {\bf X}{\bf X}^{\text{H}}\right)\right),
\end{align*}
where ${\bf H}$ is given by (\ref{eq_lemma}).
\end{proof}      \vspace{2mm}

With Lemmata \ref{lemmaTrace} and \ref{lemmaL21norm}, we have the following theorem.
\begin{theorem}
\label{nonincreasingNoiseless}
The sequence $\{{\bf Z}_{k}\}$ generated by ${\bf Z}_{k+1} = \lambda({\bf P}_{k} + \lambda{\bf Q}_{k})^{-1}{\bf Q}_{k}{\bf Y}$ produces a non-increasing objective function defined in (\ref{prob_noiseless}), i.e., $f({\bf Z}_{k}) \geq f({\bf Z}_{k+1})$ for $k = 0, 1, 2, \cdots$. {\color{black}Moreover, the sequence $\{{\bf Z}_{k}\}$ is bounded, and $\lim_{k \to \infty} \|{\bf Z}_{k} - {\bf Z}_{k+1}\|_{\mathrm{F}} = 0$.}
\end{theorem}   \vspace{-6mm}

{\color{black}
\begin{proof}
See Appendix \ref{proof_nonincreasingNoiseless}.
\end{proof}
}

\begin{theorem}
\label{boundedNoiseless}
The objective function $f({\bf Z}) = \|[{\bf Z}, \mu{\bf I}]\|_{*} + \lambda\|[{\bf Y} - {\bf Z}, \mu{\bf 1}]\|_{2,1}$ is bounded below by $|\mu|(\sqrt{M} + \lambda M)$.
\end{theorem} \vspace{-6mm}

\begin{proof}
See Appendix \ref{proof_boundedNoiseless}.
\end{proof}


{\color{black}
\begin{theorem}
\label{KKTpointNoiseless}
Any limit point of the sequence $\{{\bf Z}_{k}\}$ generated by (\ref{irls_progress_noiseless}) is a stationary point of Problem (\ref{prob_noiseless}), and moreover, the stationary point is globally optimal.
\end{theorem} \vspace{-6mm}

\begin{proof}
See Appendix \ref{proof_KKTpointNoiseless}.
\end{proof}
}

\subsection{NOISY CASE}
\label{noisycase}
In the noisy case, the data model is as (\ref{datamodel_ZV}), and the problem to be solved is given as 
\begin{align}
\label{problem_nuclearl1}
\min_{{\bf Z}, {\bf V}} ~ \frac{1}{2}\|{\bf Y} \! - \! {\bf Z} \! - \! {\bf V}\|_{\text{F}}^{2} + \lambda_{1}\|{\bf Z}\|_{*} + \lambda_{2}\|{\bf V}\|_{2,1},
\end{align}
where $\lambda_{1}$ and $\lambda_{2}$ are two tuning parameters. Different from the noiseless case, we have to optimize the problem with two variables, i.e., ${\bf Z}$ and ${\bf V}$. To proceed, we also introduce a smoothing parameter $\mu$ into the nuclear norm and the $\ell_{2,1}$ mixed-norm in Problem (\ref{problem_nuclearl1}). Therefore, the problem to be addressed is transferred to
\begin{align}
\label{prob_f_ZV}
\min_{{\bf Z}, {\bf V}} ~ f({\bf Z} , {\bf V}),
\end{align}
where the objective function is defined as $f({\bf Z} , {\bf V}) \triangleq \frac{1}{2}\|{\bf Y} \! - \! {\bf Z} \! - \! {\bf V}\|_{\text{F}}^{2} + \lambda_{1}\|[{\bf Z}, \mu{\bf I}]\|_{*} + \lambda_{2}\|[{\bf V}, \mu{\bf 1}]\|_{2,1}$. The derivatives of $f({\bf Z}, {\bf V})$ w.r.t. ${\bf Z}$ and ${\bf V}$ are 
\begin{align*}
\frac{\partial f({\bf Z} , {\bf V})}{\partial {\bf Z}} & = (-{\bf Y} + {\bf Z} + {\bf V}) + \lambda_{1}{\bf P}{\bf Z}, \\
\frac{\partial f({\bf Z} , {\bf V})}{\partial {\bf V}} & = (-{\bf Y} + {\bf Z} + {\bf V}) + \lambda_{2}{\bf Q}{\bf V},
\end{align*}
respectively, where ${\bf P}$ is defined the same as that in the noiseless case, and
\begin{align}
\label{Q_noisy}
{\bf Q} \triangleq \left[ \!\!
\begin{array}{ccc}
\frac{1}{\sqrt{\|{\bf V}_{1,:}\|_{2}^{2} + \mu^{2}}} & & \\
 & \!\!\!\! \ddots \!\!\!\! & \\
 & & \frac{1}{\sqrt{\|{\bf V}_{M,:}\|_{2}^{2} + \mu^{2}}}
\end{array}  \!\! \right] \! .
\end{align}
Note that ${\bf Q}$ given in (\ref{Q_noiseless}) is exactly the same as the one in (\ref{Q_noisy}) since ${\bf V} = {\bf Y} \! - \! {\bf Z}$ in the noiseless case.

According to the KKT condition, we have
\begin{align*}
\left\{ \!\!
\begin{array}{r}
({\bf I} + \lambda_{1}{\bf P}){\bf Z} - {\bf Y} + {\bf V} = {\bf 0} \\
({\bf I} + \lambda_{2}{\bf Q}){\bf V} - {\bf Y} + {\bf Z} = {\bf 0}
\end{array} \right.
\end{align*}
which leads to the IRLS procedure as
\begin{align}
\label{irls_progress_noisy}
\left\{ \!\!
\begin{array}{l}
{\bf Z}_{k+1} = ({\bf I} + \lambda_{1}{\bf P}_{k})^{-1}({\bf Y} - {\bf V}_{k}) \\
{\bf V}_{k+1} \! = ({\bf I} + \lambda_{2}{\bf Q}_{k})^{-1}({\bf Y} - {\bf Z}_{k+1}),
\end{array}
\right.
\end{align}
where ${\bf P}_{k}$ and ${\bf Q}_{k}$ are dependent on ${\bf Z}_{k}$ and ${\bf V}_{k}$, respectively. The IRLS algorithm for the noisy case is summarized in Algorithm \ref{IRLS_noisy}.

\begin{algorithm}[t]
	\caption{IRLS algorithm for noisy case}
	\label{IRLS_noisy}
	\hspace*{\algorithmicindent} \textbf{Input~~~~\!:} ${\bf Y} \in \mathbb{C}^{M \times T}$, $\lambda_{1}$, $\lambda_{2}$, $\mu$, $\epsilon$, $k_{\text{max}}$ \\
	\hspace*{\algorithmicindent} \textbf{Output~~\!:} ${\bf{\widehat Z}} \in \mathbb{C}^{M \times T}$, ${\bf{\widehat V}} \in \mathbb{C}^{M \times T}$ \\
	\hspace*{\algorithmicindent} \textbf{Initialize:} ${\bf Z}_{0} \gets {\bf Z}_{\text{init}}$, ${\bf V}_{0} \gets {\bf V}_{\text{init}}$, $k \gets 0$
	\begin{algorithmic}[1]
		\While {not converged}
		\State $k \gets k+1$
		\State calculate ${\bf P}_{k}$ and ${\bf Q}_{k}$
		\State update ${\bf Z}_{k}$ using ${\bf Z} = ({\bf I} + \lambda_{1}{\bf P})^{-1}({\bf Y} - {\bf V})$
		\State update ${\bf V}_{k}$ using ${\bf V} = ({\bf I} + \lambda_{2}{\bf Q})^{-1}({\bf Y} - {\bf Z})$
		\State converged $\gets$ $\left\{ \!\! \begin{array}{l}
		k \geq k_{\text{max}} ~ \text{or} \\
		{\color{black} \frac{ | f({\bf Z}_{k} , \! {\bf V}_{k}) - f({\bf Z}_{k-1} , \! {\bf V}_{k-1}) | }{ | f({\bf Z}_{k} , \! {\bf V}_{k}) | } \leq \epsilon }
		\end{array} \right. $
		\EndWhile \State \textbf{end while} \\
		${\bf{\widehat Z}} \gets {\bf Z}_{k}$, ${\bf{\widehat V}} \gets {\bf V}_{k}$
	\end{algorithmic}
\end{algorithm}

\subsection{CONVERGENCE ANALYSIS FOR ALGORITHM \ref{IRLS_noisy}}
\label{convergence_noisy}
In this part, the monotonicity and boundedness of the objective function $f({\bf Z}, {\bf V})$ in (\ref{prob_f_ZV}) are proved in Theorems \ref{nonincreasingNoisy} and \ref{boundedNoisy}, respectively.

\begin{theorem}
\label{nonincreasingNoisy}
The sequence $\{({\bf Z}_{k} , {\bf V}_{k})\}$ generated by (\ref{irls_progress_noisy}) produces a non-increasing objective function defined in (\ref{prob_f_ZV}), i.e., $f({\bf Z}_{k} , {\bf V}_{k}) \geq f({\bf Z}_{k+1} , {\bf V}_{k+1})$ for $k = 0, 1, 2, \cdots$. {\color{black}Moreover, the sequence $\{({\bf Z}_{k} , {\bf V}_{k})\}$ is bounded, and $\lim_{k \to \infty}\|{\bf Z}_{k} - {\bf Z}_{k+1}\|_{\mathrm{F}} = 0$ and $\lim_{k \to \infty}\|{\bf V}_{k} - {\bf V}_{k+1}\|_{\mathrm{F}} = 0$.}
\end{theorem}    \vspace{-6mm}

{\color{black}
\begin{proof}
See Appendix \ref{proof_nonincreasingNoisy}.
\end{proof}
}

\begin{theorem}
\label{boundedNoisy}
The objective function $f({\bf Z} , \! {\bf V}) \! = \! \frac{1}{2}\|{\bf Y} \! - \! {\bf Z} \! - \! {\bf V}\|_{\rm{F}}^{2} \! + \! \lambda_{1}\|[{\bf Z}, \mu{\bf I}]\|_{*} \! + \! \lambda_{2}\|[{\bf V} , \mu{\bf 1}]\|_{2,1}$ is bounded below by $|\mu|(\lambda_{1}\sqrt{M} + \lambda_{2}M)$.
\end{theorem}   \vspace{-6mm}

\begin{proof}
See Appendix \ref{proof_boundedNoisy}.
\end{proof}

{\color{black}
\begin{theorem}
\label{KKTpointNoisy}
Any limit point of the sequence $\{({\bf Z}_{k}, {\bf V}_{k})\}$ generated by (\ref{irls_progress_noisy}) is a stationary point of Problem (\ref{prob_f_ZV}), and moreover, the stationary point is globally optimal.
\end{theorem} \vspace{-6mm}

\begin{proof}
See Appendix \ref{proof_KKTpointNoisy}.
\end{proof}
}

\begin{remark}
The differences between our work and \cite{Lu2015(b)} are stated as follows.
\begin{itemize}
\item The problem formulation in \cite{Lu2015(b)} is column-sparse, while we have row-sparsity of ${\bf V}$. This leads to differences in matrix multiplication and matrix derivative.
\item \cite{Lu2015(b)} considers the noiseless case only, while we consider both noiseless and noisy cases. 
\item To update ${\bf Z}$ using matrices ${\bf P}$ and ${\bf Q}$, the approach in \cite{Lu2015(b)} involves a Sylvester equation and utilizes the Matlab command \texttt{lyap}. However, our method admits a closed-form formula, see (\ref{irls_progress_noiseless}) and (\ref{irls_progress_noisy}).
\item The proofs of convergence are not exactly the same. \cite{Lu2015(b)} proves the monotonicity of the objective and the boundedness of the sequence $\{{\bf Z}_{k}\}$. We prove the monotonicity and the boundedness of the objective in both noiseless and noisy cases.
\end{itemize}
\end{remark}

\subsection{DOA ESTIMATION AND DISTORTED SENSOR DETECTION}
\label{estimationanddetection}
Once ${\widehat{\bf Z}}$ and $\widehat{\bf V}$ are resolved, they can be adopted to estimate the DOAs and detect the distorted sensors, respectively. Note that ${\bf Z} = {\bf A}{\bf S}$ can be viewed as a noise-free data model. DOAs can be found via subspace-based methods, such as MUSIC, whose spatial spectrum is
\begin{align*}
P(\theta) = \frac{1}{{\bf a}^{\text{H}}(\theta)({\bf I} - {\bf L}{\bf L}^{\text{H}}){\bf a}(\theta)}.
\end{align*}
The SVD of $\widehat{\bf Z}$ is $\widehat{\bf Z} = {\bf L}{\bm \Sigma}{\bf R}^{\text{H}}$, where the columns of ${\bf L}$ and ${\bf R}$ contain the left and right orthogonal base vectors of $\widehat{\bf Z}$, respectively, and ${\bm \Sigma}$ is a diagonal matrix whose diagonal elements are the singular values of $\widehat{\bf Z}$ arranged in descending order. Under the assumption that the number of sources, i.e., $K$, is known, the DOAs are determined by searching for the $K$ largest peaks of $P(\theta)$.

On the other hand, the number of distorted sensors and their positions can be determined by $\|\widehat{\bf V}_{i,:}\|_{2}$, $i = 1, 2, \cdots, M$. Algorithm \ref{detec_failsensor} shows a strategy for detecting the distorted sensors. In words, we first calculate the $\ell_{2}$ norm of each row of $\widehat{\bf V}$ and form a vector, say ${\bf v}$, and then we sort these $\ell_{2}$ norms in ascending order and obtain $\tilde{\bf v}$. We define the difference of the first two entries of $\tilde{\bf v}$ as $d = \tilde{\bf v}(2) - \tilde{\bf v}(1)$. Next, for $i = 3, 4, \cdots, M$, we compute $\tilde{\bf v}(i) - \tilde{\bf v}(i - 1)$ and compare it with a threshold, say $h$, of large value: if it is larger than or equal to $h$, we set $i_{\text{fail}} = i$ and break the for loop; if it is less than $h$, we have $i_{\text{fail}} = M + 1$. Finally, the number of distorted sensors is obtained as $M_{\text{fail}} = M - i_{\text{fail}} +1$.

\begin{algorithm}[t]
	\caption{Detection of distorted sensors}
	\label{detec_failsensor}
	\hspace*{\algorithmicindent} \textbf{Input~~:} $\widehat{\bf V} \in \mathbb{C}^{M \times T}$, $h$  \\
	 \hspace*{\algorithmicindent} \textbf{Output:} $M_{\text{fail}}$ \\
	\hspace*{\algorithmicindent} calculate ${\bf v} = [\|\widehat{\bf V}_{1,:}\|_{2}, \|\widehat{\bf V}_{2,:}\|_{2}, \cdots, \|\widehat{\bf V}_{M,:}\|_{2} ]^{\text{T}}$ \\
	\hspace*{\algorithmicindent} calculate $\tilde{\bf v} = \text{sort}({\bf v} , \text{`ascend'})$ \\
	\hspace*{\algorithmicindent} calculate $d = \tilde{\bf v}(2) - \tilde{\bf v}(1)$ and assign $i_{\text{fail}} = M + 1$
	\begin{algorithmic}[1]
		\For {$i = 3, 4, \cdots, M$}
		\If {$\tilde{\bf v}(i) - \tilde{\bf v}(i - 1) \geq h$}
		\State $i_{\text{fail}} = i$  and break the \textbf{for} loop
		\EndIf \State \textbf{end if}
		\EndFor \State \textbf{end for} \\
		$M_{\text{fail}} \gets M \! - \! i_{\text{fail}} +1$
	\end{algorithmic}
\end{algorithm}

\section{SIMULATIONS}
\label{simulation}

\subsection{PARAMETER SELECTION}
\label{parameterselection}
In this subsection, we discuss the problem of choosing appropriate values for $\mu$, $\lambda_{1}$, and $\lambda_{2}$ in Problem (\ref{prob_f_ZV}) used in Algorithm \ref{IRLS_noisy}. We set $\epsilon = 10^{-16}$, $k_{\text{max}} = 1000$, and ${\bf Z}_{\text{init}} = {\bf V}_{\text{init}} = {\bf O}$ in Algorithm \ref{IRLS_noisy}, where ${\bf O}$ denotes the $M \times T$ all-zeros matrix. We define the root-mean squared error (RMSE) of DOA estimates as:
\begin{align*}
\text{RMSE} = \sqrt{\frac{1}{QK}\sum_{q=1}^{Q}\sum_{k=1}^{K} (\hat{\theta}_{k,q} - {\theta}_{k})^{2} },
\end{align*}
where $\hat{\theta}_{k,q}$ is the estimate of the $k$th signal in the $q$th Monte Carlo trial, and $Q$ is the total number of Monte Carlo trials. The RMSE is used as a metric to select appropriate values for $\mu$, $\lambda_{1}$, and $\lambda_{2}$. The plots in this subsection are averaged over $Q = 1000$ trials.

Consider a uniform linear array (ULA) of $M = 10$ sensors, $4$ of which at random positions are distorted by gain and phase errors, receiving $K = 2$ signals with DOAs ${\bm \theta} = [-10^{\circ}, 10^{\circ}]^{\text{T}}$. The sensor gain and phase errors are randomly generated by drawing from uniform distributions on $[0, 10]$ and $[-15^{\circ}, 15^{\circ}]$, respectively. In the first example, we test $6$ scenarios with different signal-to-noise ratios (SNRs) and different numbers of snapshots. In Fig. \ref{RMSEvsMU}, we fix $\lambda_{1} = 2$ and $\lambda_{2} = 0.2$, and plot RMSE versus $\mu$. In the second example, we examine RMSE versus the tuning parameters $\lambda_{1}$ and $\lambda_{2}$ with $\mu = 0.01$, SNR = $0$ dB, and $T = 100$ snapshots. The result is drawn in Fig. \ref{RMSEvsLambdas}.

We observe from Fig. \ref{RMSEvsMU} that the RMSE remains unchanged and stays minimal when $\mu$ lies within the interval $[10^{-13} , 10^{0}]$ for all $6$ tested scenarios. Hence, we can choose any value for $\mu$ within this interval. Since the interval covers such a large range, the IRLS algorithm is insensitive to the smoothing parameter $\mu$. Note that in Fig. \ref{RMSEvsLambdas}, our goal is to find a pair of $(\lambda_{1}, \lambda_{2})$ such that the RMSE is minimized. This demonstrates that there are many pairs of $(\lambda_{1}, \lambda_{2})$ meeting such a condition, such as $(\lambda_{1}, \lambda_{2}) = (2 ~\! , 0.2)$, which is used for Algorithm \ref{IRLS_noisy} in the following simulations.

\begin{figure}[t]
	\centering
	\subfigure{\includegraphics[width=0.5\textwidth]{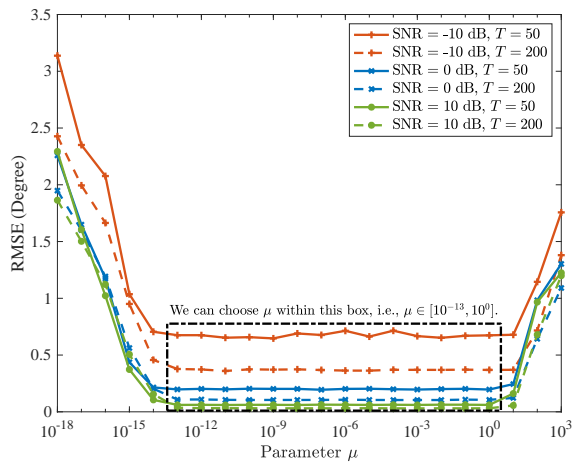}} \vspace{-4mm}
	\caption{RMSE versus $\mu$, with $M = 10$ sensors (4 of which fail), $K = 2$ sources, $\lambda_{1} = 2$, and $\lambda_{2} = 0.2$.}
	\label{RMSEvsMU}
\end{figure} 

\begin{figure}[t]
	\centering
	\subfigure{\includegraphics[width=0.5\textwidth]{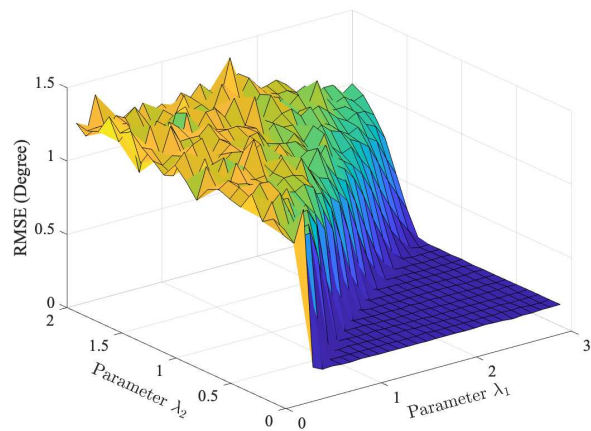}} \vspace{-4mm}
	\caption{RMSE versus $\lambda_{1}$ and $\lambda_2$, with $M = 10$ sensors (4 of which fail), $K = 2$ sources, $T = 100$ snapshots, $\text{SNR} = 0 ~ \text{dB}$, and $\mu = 0.01$.}
	\label{RMSEvsLambdas}
\end{figure}

\subsection{CONVERGENCE SPEED}
\label{convergencespeed}
We compare the convergence speed of the IRLS with several existing methods, i.e., SVT, APG, and ADMM. Considering again a ULA of $M = 10$ sensors, $4$ of which at random positions are distorted, receives $K = 2$ signals from $-10^{\circ}$ and $10^{\circ}$. The objective function values of the algorithms versus the number of iterations are depicted in Fig. \ref{ObjecFunvsNumofIter} with SNR $ = 0$ dB and $T = 100$ snapshots. We see that the IRLS algorithm converges fastest in the sense that its objective function value decreases most rapidly, and it requires the least number of iterations to terminate, compared with the other three competitors.

The objective function value, CPU time and number of iterations are tabulated in Table \ref{table_objecfun_time_iter} (upper) for SNR $ = 0$ dB and $T = 100$ snapshots, and Table \ref{table_objecfun_time_iter} (lower) for SNR $ = 0$ dB and $T = 500$ snapshots. In both settings, the IRLS algorithm has the smallest objective function value, the least CPU time, and the least number of iterations, among all the examined algorithms.

\begin{figure}[t]
	\centering
	\subfigure{\includegraphics[width=0.5\textwidth]{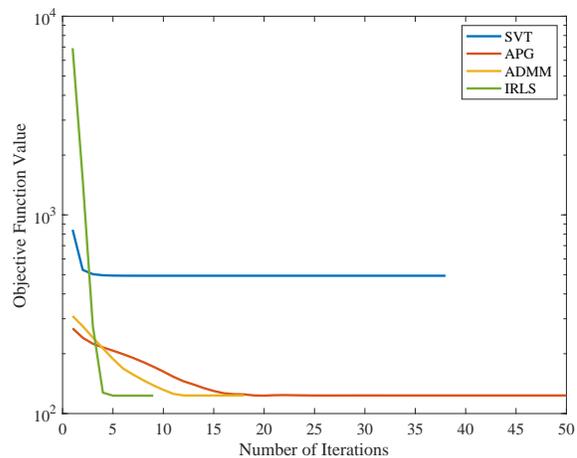}} \vspace{-4mm}
	\caption{Objective function value versus number of iterations at $\text{SNR} = 0 ~ \text{dB}$ and $T = 100$ snapshots.}
	\label{ObjecFunvsNumofIter}
\end{figure}

\begin{table}[t]
\centering
\caption{Comparison of objective function value, CPU time and number of iterations in two different settings.}
\label{table_objecfun_time_iter}
   \begin{tabular}{c|c|c|c}
    \multicolumn{4}{c}{$T = 100$ snapshots, $\text{SNR} = 0 ~ \text{dB}$} \\ \hline
    Algorithm & $f({\bf Z} , {\bf V})$ & Time (sec) & No. Iter. \\ \hline
    {\color{black} SVT \cite{Cai2010}} & 493.8653 & 0.1205 & {\color{black}{\textbf{5}}} \\ \hline
    {\color{black}APG \cite{Beck2009}} & 123.3227 & 0.8403 & {\color{black}22} \\ \hline
    {\color{black}ADMM \cite{Lin2013} }& 123.3227 & 0.1029 & {\color{black}13} \\ \hline
    IRLS & {\textbf{123.0227}} & {\textbf{0.0890}} & {\color{black}{\textbf{5}}}
   \end{tabular} \\ \vspace{2mm}
   \begin{tabular}{c|c|c|c}
    \multicolumn{4}{c}{$T = 500$ snapshots, $\text{SNR} = 0 ~ \text{dB}$} \\ \hline
    Algorithm & $f({\bf Z} , {\bf V})$ & Time (sec) & No. Iter. \\ \hline
    {\color{black}SVT \cite{Cai2010}} & \!\!\!1965.5994 & 0.6166 & {\color{black}25} \\ \hline
    {\color{black}APG \cite{Beck2009}} & 282.5776 & 4.3472 & {\color{black}32} \\ \hline
    {\color{black}ADMM \cite{Lin2013}} & 282.5776 & 5.8497 & {\color{black}78} \\ \hline
    IRLS & {\textbf{282.2776}} & {\textbf{0.1063}} & {\color{black}{\textbf{10}}}
   \end{tabular}
\end{table}

\subsection{COMPUTATIONAL COMPLEXITY}
\label{complexity}
We compare the computational complexity in this subsection. Note that the SVT, APG, and ADMM algorithms require one SVD of an $M \! \times \! T$ matrix per iteration, and the SVD consumes the most CPU time. As for the IRLS algorithm, the main calculation is to find the inverse of an $M \! \times \! M$ matrix per iteration. Their main computational cost is summarized in Table \ref{table_cost}, where $K_{\text{svt}}$, $K_{\text{apg}}$, $K_{\text{admm}}$, and $K_{\text{irls}}$ denote the numbers of iterations for the SVT, APG, ADMM, and IRLS algorithms, respectively.

In Fig. \ref{costvssnapshot}, we plot the averaged CPU time against the number of snapshots at $M = 10$ sensors (4 of which distorted), $K = 2$ sources, SNR $= 0$ dB, and $Q = 1000$ Monte Carlo runs. It is seen that the CPU times of the SVT, APG, and ADMM\footnote{Note that there is a jump of ADMM at $T = 250$. This is caused by the rapid increment of its number of iterations $K_{\text{admm}}$.} algorithms are nearly linearly increasing with the number of snapshots. This is consistent with the theoretical analysis in Table \ref{table_cost}. Fig. \ref{costvssensor} displays the CPU time versus the number of sensors with $T = 100$ snapshots and the other parameters are the same as those in Fig. \ref{costvssnapshot}. We see that the curves of the CPU time of the SVT, APG, and ADMM algorithms are approximately linearly correlated to the number of sensors in a log scale, which again matches the theoretical calculations in Table \ref{table_cost}.

\begin{figure}[t]
	\centering
	\subfigure{\includegraphics[width=0.5\textwidth]{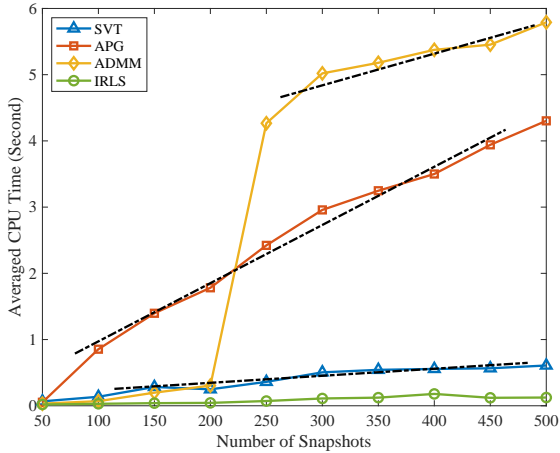}} \vspace{-4mm}   
	\caption{Computational complexity versus number of snapshots.}
	\label{costvssnapshot}
\end{figure}

\begin{figure}[t]
	\centering
	\subfigure{\includegraphics[width=0.5\textwidth]{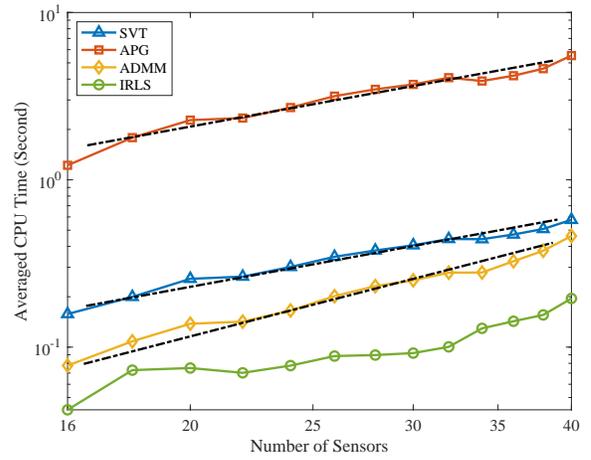}} \vspace{-4mm}   
	\caption{Computational complexity versus number of sensors.}
	\label{costvssensor}
\end{figure}

\begin{table}[t]
\centering
\caption{Computational complexity.} 
\label{table_cost}
	\begin{tabular}{c|c}
	Algorithm & Complexity \\ \hline
	{\color{black}SVT \cite{Cai2010}} & $K_{\text{svt}}\mathcal{O}(T\!M^{2})$ \\ \hline
	{\color{black}APG \cite{Beck2009}} & $K_{\text{apg}}\mathcal{O}(T\!M^{2})$ \\ \hline
	{\color{black}ADMM \cite{Lin2013}} & $K_{\text{admm}}\mathcal{O}(T\!M^{2})$ \\ \hline
	IRLS & $K_{\text{irls}}\mathcal{O}(M^{3})$
	\end{tabular}
\end{table}

\subsection{DOA ESTIMATION PERFORMANCE}
\label{doaperformance}
We use the RMSE and resolution probability as DOA estimation performance measures. The resolution probability is defined as
\begin{align*}
\text{ResProb} = {N_{\text{succ}}}/{Q},
\end{align*}
where as in the previous examples $Q$ is the number of Monte Carlo runs, and $N_{\text{succ}}$ denotes the number of trials where all the DOAs are successfully estimated. The trial is counted as a successful one if the following inequality is satisfied: $\max_{k}\{ |\hat{\theta}_{k} - \theta_{k}| \} \leq 0.5^{\circ}$.

In the first example, we consider a ULA of $M = 10$ sensors, $3$ of which at random positions are distorted, $K = 2$ signals from $-10^{\circ}$ and $10^{\circ}$, $T = 100$ snapshots, and $Q = 5000$ Monte Carlo trials. The RMSE and resolution probability are depicted in Figs. \ref{RMSEvsSNR} and \ref{ResProbvsSNR}, respectively. {\color{black}The traditional Cramér–Rao bound (CRB) with known sensor errors \cite{Delmas2014} is plotted as a benchmark.} Note that the curve labelled as ``MUSIC-Known'' denotes the MUSIC method with exact knowledge of the distorted sensors. It is seen that the SVT and MUSIC have bad performance even when the SNR becomes large. The APG, ADMM, and IRLS algorithms perform well when the SNR increases, their RMSEs decrease and their resolution probabilities increase up to $1$. The IRLS algorithm outperforms the other two state-of-the-art methods, i.e., APG and ADMM.

\begin{figure}[t]
	\centering
	\subfigure{\includegraphics[width=0.5\textwidth]{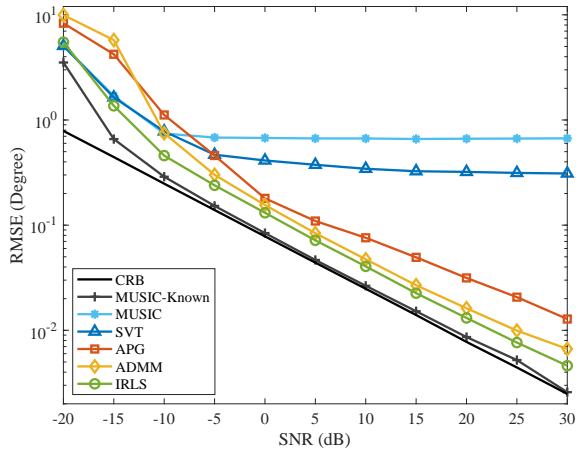}}  \vspace{-4mm}
	\caption{RMSE versus SNR.}
	\label{RMSEvsSNR}
\end{figure}

\begin{figure}[t]
	\centering
	\subfigure{\includegraphics[width=0.5\textwidth]{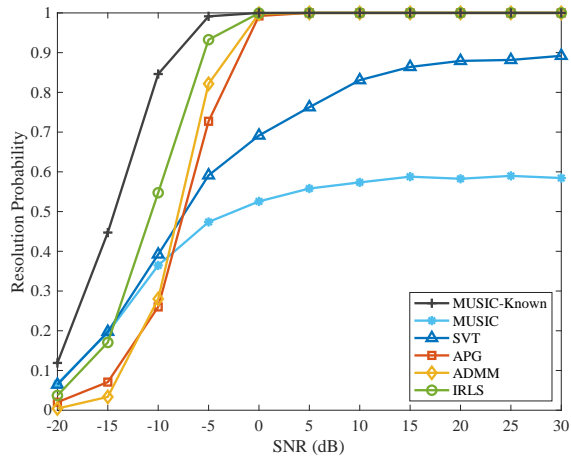}} \vspace{-4mm}
	\caption{Resolution probability versus SNR.}
	\label{ResProbvsSNR}
\end{figure}

In the next example, we examine the DOA estimation performance for different {\color{black}numbers of snapshots}. The SNR is set to be $0$ dB, and the remaining parameters are the same as those of the former example. The RMSE and resolution probability of the methods are plotted in Figs. \ref{RMSEvsSnapshot} and \ref{ResProbvsSnapshot}, respectively. The results demonstrate a better performance of the IRLS algorithm compared with the SVT, APG, and ADMM methods.

\begin{figure}[t]
	\centering
	\subfigure{\includegraphics[width=0.5\textwidth]{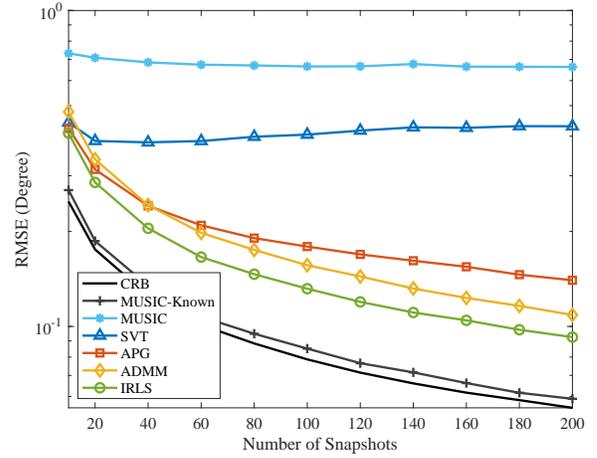}}  \vspace{-4mm}
	\caption{RMSE versus number of snapshots.}
	\label{RMSEvsSnapshot}
\end{figure}

\begin{figure}[t]
	\centering
	\subfigure{\includegraphics[width=0.5\textwidth]{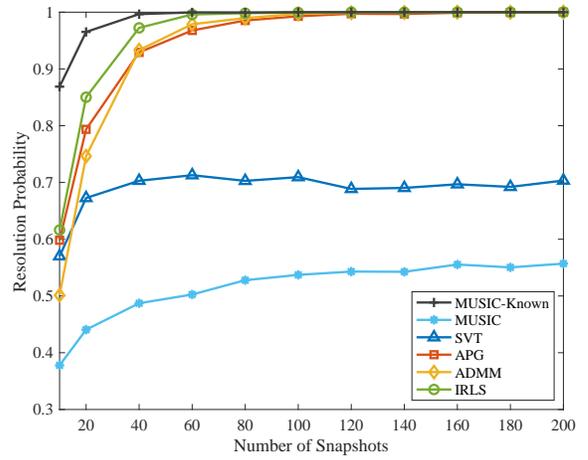}} \vspace{-4mm}
	\caption{Resolution probability versus number of snapshots.}
	\label{ResProbvsSnapshot}
\end{figure}

In the last example of this subsection, we evaluate the DOA estimation performance in view of the source separation angle. The settings of SNR $= 0$ dB, $K = 2$ sources, and $T = 100$ snapshots are employed. The first signal is from $0^{\circ}$, while the DOA of the second signal changes from $1^{\circ}$ to $20^{\circ}$ with a stepsize of $1^{\circ}$. The other parameters are unchanged as those in the first example of this subsection. The RMSE and resolution probability versus angular separation are displayed in Figs. \ref{RMSEvsAngSepa} and \ref{ResProbvsAngSepa}, respectively. These again indicate that the IRLS algorithm outperforms the SVT, APG, and ADMM algorithms in terms of RMSE and resolution probability.

\begin{figure}[t]
	\centering
	\subfigure{\includegraphics[width=0.5\textwidth]{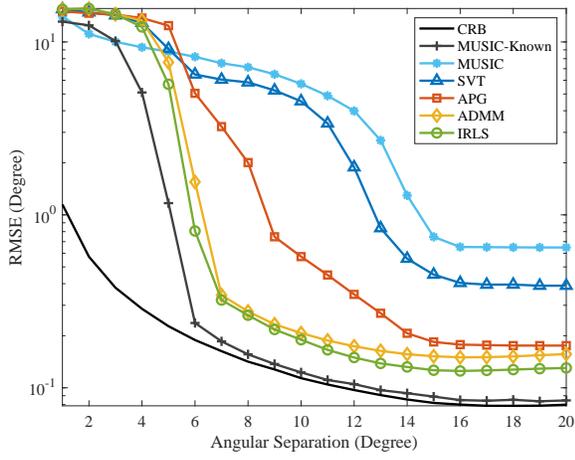}}  \vspace{-4mm}
	\caption{RMSE versus source separation angle.}
	\label{RMSEvsAngSepa}
\end{figure}

\begin{figure}[t]
	\centering
	\subfigure{\includegraphics[width=0.5\textwidth]{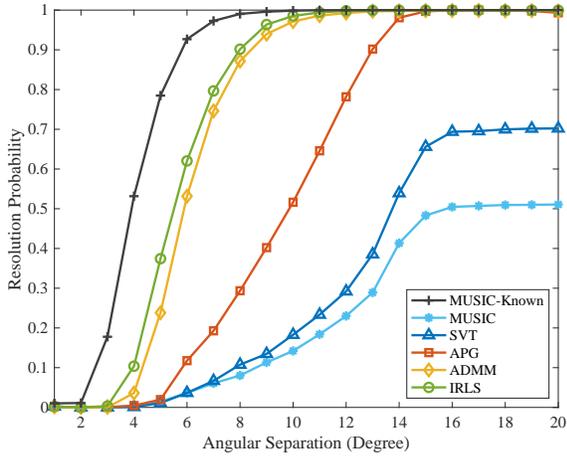}} \vspace{-4mm}
	\caption{Resolution probability versus source separation angle.}
	\label{ResProbvsAngSepa}
\end{figure}

\subsection{DISTORTED SENSOR DETECTION PERFORMANCE}
\label{detectionperformance}
Parallel to the three examples in Section \ref{doaperformance}, we now examine the performance of the detection of distorted sensors of the SVT, APG, ADMM, and IRLS algorithms. The threshold in Algorithm \ref{detec_failsensor} is set as $h = 10d$. We utilize the success detection rate as a metric, which is defined as 
\begin{align*}
\text{DetecRate} = {N_{\text{detec}}}/{Q}.
\end{align*}
$N_{\text{detec}}$ denotes the number of trials where the number of distorted sensors is correctly estimated, and meanwhile their positions are exactly found. The results are given in Figs. \ref{RatevsSNR}, \ref{RatevsSnapshot}, and \ref{RatevsAngSepa}, which show that the ADMM is the best amongst all tested methods in terms of identifying the distorted sensors, followed by the IRLS algorithm.

\begin{figure}[t]
	\centering
	\subfigure{\includegraphics[width=0.5\textwidth]{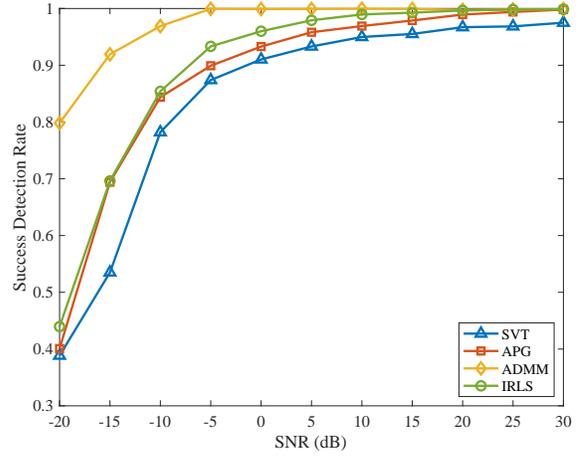}} \vspace{-4mm}
	\caption{Success detection rate versus SNR.}
	\label{RatevsSNR}
\end{figure}

\begin{figure}[t]
	\centering
	\subfigure{\includegraphics[width=0.5\textwidth]{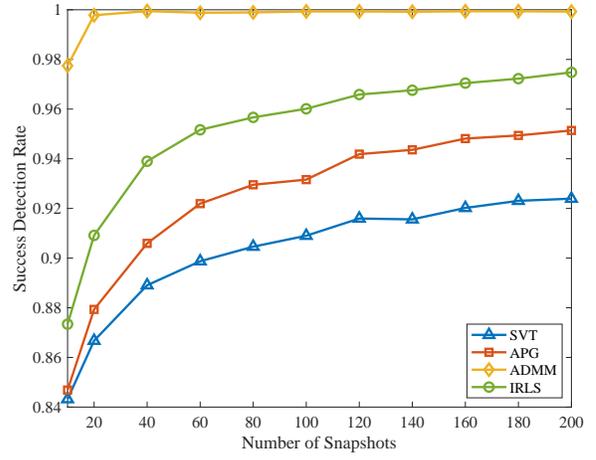}} \vspace{-4mm}
	\caption{Success detection rate versus number of snapshots.}
	\label{RatevsSnapshot}
\end{figure}

\begin{figure}[t]
	\centering
	\subfigure{\includegraphics[width=0.5\textwidth]{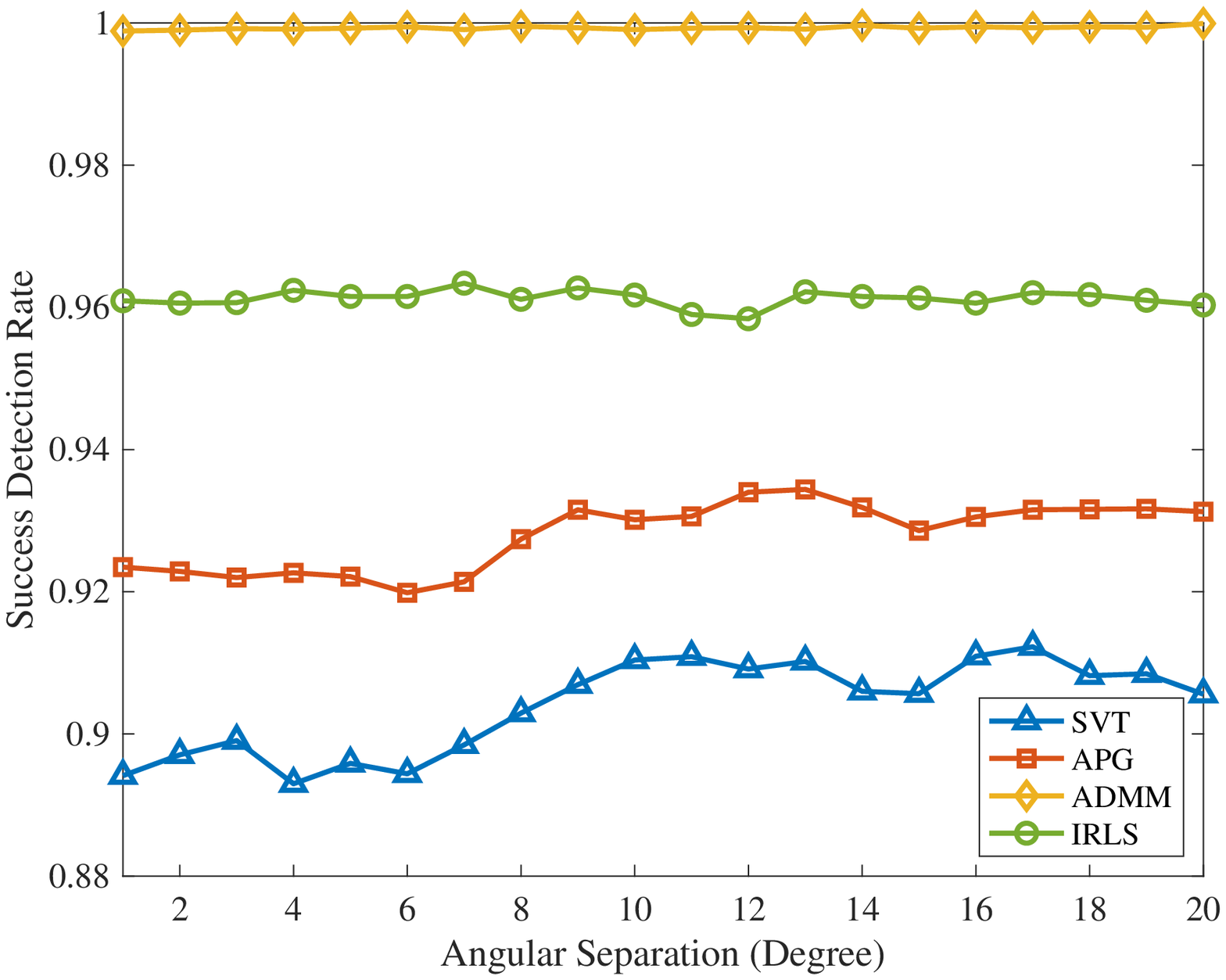}}  \vspace{-4mm}
	\caption{Success detection rate versus source separation angle.}
	\label{RatevsAngSepa}
\end{figure}

\section{CONCLUSION}
\label{conclusion}

We studied the problem of simultaneously estimating direction-of-arrival (DOA) of signals and detecting distorted sensors. It is assumed that the distorted sensors occur randomly, and the number of distorted sensors is much smaller than the total number of sensors. The problem was formulated via low-rank and row-sparse decomposition, and solved by iteratively reweighted least squares (IRLS). Both noiseless and noisy cases were considered. Theoretical analyses of algorithm convergence were provided. Computational cost of the IRLS algorithm was compared with that of several existing methods. Simulation results were conducted for parameter selection, convergence speed, computational time, and performance of DOA estimation as well as distorted sensor detection. The IRLS method was demonstrated to have higher DOA estimate accuracy and lower computational cost than other methods, and the alternating direction method of multipliers was shown to be slightly better than the IRLS algorithm in distorted sensor detection.

\section*{APPENDIX}

\subsection{Proof of Theorem \ref{nonincreasingNoiseless}}
\label{proof_nonincreasingNoiseless}
We calculate the difference between the objective function values in two successive iterations as 
\begingroup
\allowdisplaybreaks
\begin{align}
& f({\bf Z}_{k}) - f({\bf Z}_{k+1}) \nonumber \\
= ~ & \|[{\bf Z}_{k}, \mu{\bf I}]\|_{*} - \|[{\bf Z}_{k+1}, \mu{\bf I}]\|_{*} \nonumber \\ 
& + \lambda \left( \|[{\bf Y} - {\bf Z}_{k}, \mu{\bf 1}]\|_{2,1} - \|[{\bf Y} - {\bf Z}_{k+1}, \mu{\bf 1}]\|_{2,1} \right) \nonumber \\
= ~ & \text{trace} \! \left( \! \left({\bf Z}_{k}{\bf Z}_{k}^{\text{H}} + \mu^{2}{\bf I}\right)^{\frac{1}{2}} \! \right) - \text{trace} \! \left( \! \left({\bf Z}_{k+1}{\bf Z}_{k+1}^{\text{H}} + \mu^{2}{\bf I}\right)^{\frac{1}{2}} \! \right) \nonumber \\
& + \lambda \left( \|[{\bf Y} - {\bf Z}_{k}, \mu{\bf 1}]\|_{2,1} - \|[{\bf Y} - {\bf Z}_{k+1}, \mu{\bf 1}]\|_{2,1} \right) \nonumber \\
\label{ineq_1}
\geq ~ & \text{trace} \! \left( \frac{1}{2}\left({\bf Z}_{k}{\bf Z}_{k}^{\text{H}} - {\bf Z}_{k+1}{\bf Z}_{k+1}^{\text{H}}\right){\bf P}_{k} \! \right) \nonumber \\
& + \frac{\lambda}{2} \text{trace} \! \left( {\bf Q}_{k} \left[ \left({\bf Y} \! - \! {\bf Z}_{k}\right)\left({\bf Y} \! - \! {\bf Z}_{k}\right)^{\text{H}}  \right. \right. \nonumber \\
& \qquad \qquad \quad  \left. \left. -  \left({\bf Y} \! - \! {\bf Z}_{k+1}\right)\left({\bf Y} \! - \! {\bf Z}_{k+1}\right)^{\text{H}} \right] \right) \\
= ~ & \text{trace} \! \left( \frac{1}{2}({\bf Z}_{k} - {\bf Z}_{k+1})({\bf Z}_{k} - {\bf Z}_{k+1})^{\text{H}}{\bf P}_{k} \right) \nonumber \\
& + \text{trace} \! \left( ({\bf Z}_{k} - {\bf Z}_{k+1}){\bf Z}_{k+1}^{\text{H}}{\bf P}_{k} \right) \nonumber \\
& + \frac{\lambda}{2}\text{trace} \! \left( 2{\bf Q}_{k}{\bf Y}({\bf Z}_{k+1} - {\bf Z}_{k})^{\text{H}} \right) \nonumber \\
& + \frac{\lambda}{2}\text{trace} \! \left( {\bf Q}_{k}({\bf Z}_{k}{\bf Z}_{k}^{\text{H}} - {\bf Z}_{k+1}{\bf Z}_{k+1}^{\text{H}}) \right) \nonumber \\
= ~ & \text{trace} \! \left( \frac{1}{2}({\bf Z}_{k} - {\bf Z}_{k+1})({\bf Z}_{k} - {\bf Z}_{k+1})^{\text{H}}{\bf P}_{k} \right) \nonumber \\
& + \text{trace} \! \left( ({\bf Z}_{k} - {\bf Z}_{k+1}){\bf Z}_{k+1}^{\text{H}}{\bf P}_{k} \right) \nonumber \\
& + \frac{\lambda}{2}\text{trace} \! \left( 2{\bf Q}_{k}{\bf Y}({\bf Z}_{k+1} - {\bf Z}_{k})^{\text{H}} \right) \nonumber \\
& + \frac{\lambda}{2}\text{trace} \! \left( {\bf Q}_{k}({\bf Z}_{k} - {\bf Z}_{k+1})({\bf Z}_{k} - {\bf Z}_{k+1})^{\text{H}} \right) \nonumber \\
& + \lambda ~\! \text{trace} \! \left( {\bf Q}_{k}({\bf Z}_{k} - {\bf Z}_{k+1}){\bf Z}_{k+1}^{\text{H}} \right) \nonumber \\
\label{ineq_2}
\geq ~ & \text{trace} \! \left( ({\bf Z}_{k} - {\bf Z}_{k+1}){\bf Z}_{k+1}^{\text{H}}{\bf P}_{k} \right) \nonumber \\
& + \text{trace} \! \left( \lambda{\bf Q}_{k}{\bf Y}({\bf Z}_{k+1} - {\bf Z}_{k})^{\text{H}} \right) \nonumber \\
& + \text{trace} \! \left( \lambda{\bf Q}_{k}({\bf Z}_{k} - {\bf Z}_{k+1}){\bf Z}_{k+1}^{\text{H}} \right)  \\
= ~ & \text{trace} \! \left( ({\bf Z}_{k} - {\bf Z}_{k+1}){\bf Z}_{k+1}^{\text{H}}({\bf P}_{k} + \lambda{\bf Q}_{k}) \right) \nonumber \\
& + \text{trace} \! \left( \lambda{\bf Q}_{k}{\bf Y}({\bf Z}_{k+1} - {\bf Z}_{k})^{\text{H}} \right) \nonumber \\
= ~ & \text{trace} \! \left( ({\bf P}_{k} + \lambda{\bf Q}_{k}){\bf Z}_{k+1}({\bf Z}_{k} - {\bf Z}_{k+1})^{\text{H}} \right) \nonumber \\ 
& + \text{trace} \! \left( \lambda{\bf Q}_{k}{\bf Y}({\bf Z}_{k+1} - {\bf Z}_{k})^{\text{H}} \right) \nonumber \\
\label{eq_7}
= ~ & \text{trace} \! \left( \lambda{\bf Q}_{k}{\bf Y}({\bf Z}_{k} - {\bf Z}_{k+1})^{\text{H}} \right) \nonumber \\
& + \text{trace} \! \left( \lambda{\bf Q}_{k}{\bf Y}({\bf Z}_{k+1} - {\bf Z}_{k})^{\text{H}} \right) \\
= ~ & 0, \nonumber
\end{align}
\endgroup
which indicates that $f({\bf Z})$ is a non-increasing function with sequence $\{{\bf Z}_{k}\}$ generated by the IRLS procedure (\ref{irls_progress_noiseless}). The first equality is based on the definition of the objective function in (\ref{prob_noiseless}), and the second equality uses $\|{\bf Z}\|_{*} = \text{trace} \left( \left({\bf Z}{\bf Z}^{\mathrm{H}}\right)^{-\frac{1}{2}} \right)$ when $M < T$. Inequality (\ref{ineq_1}) holds thanks to Lemmata \ref{lemmaTrace} and \ref{lemmaL21norm}. Inequality (\ref{ineq_2}) holds because of $\text{trace} \! \left( ({\bf Z}_{k} - {\bf Z}_{k+1})({\bf Z}_{k} - {\bf Z}_{k+1})^{\text{H}}{\bf P}\right) \geq 0$ and $\text{trace} \! \left( {\bf Q}({\bf Z}_{k} - {\bf Z}_{k+1})({\bf Z}_{k} - {\bf Z}_{k+1})^{\text{H}} \right) \geq 0$, which result from the fact that ${\bf P}$ and ${\bf Q}$ are symmetric matrices and $\text{trace} \! \left( {\bf X}{\bf X}^{\text{H}}\right) = \|{\bf X}\|_{\text{F}}^{2} \geq 0$. Equality (\ref{eq_7}) holds true according to the KKT condition, i.e., $({\bf P}_{k} + \lambda{\bf Q}_{k}){\bf Z}_{k+1} = \lambda{\bf Q}_{k}{\bf Y}$.

{\color{black}Since $f({\bf Z})$ is a non-increasing, we have 
\begin{align*}
\|{\bf Z}_{k}\|_{*} & = \text{trace} \! \left( \left({\bf Z}{\bf Z}^{\mathrm{H}}\right)^{-\frac{1}{2}} \right) < \text{trace} \! \left( \left({\bf Z}{\bf Z}^{\mathrm{H}} + \mu^{2}{\bf I}\right)^{-\frac{1}{2}} \right) \\
& = \|[{\bf Z}_{k} , \mu{\bf I}]\|_{*} \leq f({\bf Z}_{k}) \leq f({\bf Z}_{0}),
\end{align*}
which indicates that the sequence $\{{\bf Z}_{k}\}$ is bounded in terms of its nuclear norm.

Besides, combining (\ref{irls_progress_noiseless}) and (\ref{ineq_1}) yields
\begingroup
\allowdisplaybreaks
\begin{align}
& f({\bf Z}_{k}) - f({\bf Z}_{k+1}) \nonumber \\
\geq ~ & \text{trace} \! \left( \frac{1}{2}\left({\bf Z}_{k}{\bf Z}_{k}^{\text{H}} - {\bf Z}_{k+1}{\bf Z}_{k+1}^{\text{H}}\right){\bf P}_{k} \! \right) \nonumber \\
& + \frac{\lambda}{2} \text{trace} \! \left( {\bf Q}_{k} \left[ \left({\bf Y} \! - \! {\bf Z}_{k}\right)\left({\bf Y} \! - \! {\bf Z}_{k}\right)^{\text{H}}  \right. \right. \nonumber \\
& \qquad \qquad \quad  \left. \left. -  \left({\bf Y} \! - \! {\bf Z}_{k+1}\right)\left({\bf Y} \! - \! {\bf Z}_{k+1}\right)^{\text{H}} \right] \right) \nonumber \\
= ~ & \frac{1}{2} \text{trace} \! \left( ({\bf P}_{k} + \lambda{\bf Q}_{k})({\bf Z}_{k} - {\bf Z}_{k+1})({\bf Z}_{k} - {\bf Z}_{k+1})^{\mathrm{H}} \right) \nonumber \\
\label{eq_eig_trace}
\geq ~ & \frac{1}{2} \!\! \sum_{i = 1}^{M} \! \zeta_{i} \! ({\bf P}_{k} \!\!+\!\! \lambda{\bf Q}_{k}) \zeta_{M \!- i + 1} \! ( \! ({\bf Z}_{k} \!\!-\!\! {\bf Z}_{k \!+\! 1}) \! ({\bf Z}_{k} \!\!-\!\! {\bf Z}_{k \!+\! 1})^{\! \mathrm{H}} \!) \\
\geq ~ & \frac{1}{2} \zeta_{M} \! ({\bf P}_{k} \!\!+\!\! \lambda{\bf Q}_{k}) \|{\bf Z}_{k} \! - \! {\bf Z}_{k \!+\! 1}\|_{\mathrm{F}}^{2} \geq \frac{1}{2} \zeta_{\mathrm{min}}  \! \times \! \|{\bf Z}_{k} \! - \! {\bf Z}_{k \!+\! 1}\|_{\mathrm{F}}^{2} \nonumber,
\end{align}
\endgroup
where $\zeta_{i}(\cdot)$ denotes the $i$th largest eigenvalue of its input Hermitian matrix, inequality (\ref{eq_eig_trace}) follows from the fact that $\text{trace}({\bf X}{\bf Y}) \geq \sum_{i = 1}^{M} \zeta_{i}({\bf X})\zeta_{M-i+1}({\bf Y})$ holds for any two positive semi-definite matrices ${\bf X}$ and ${\bf Y} \in \mathbb{C}^{M \times M}$ \cite{Lu2015(b)}, and in the last inequality, we have defined $\zeta_{\mathrm{min}} > 0$ as the smallest eigenvalue of ${\bf P}_{k} \!+\! \lambda{\bf Q}_{k}$ over all $k$. Summing all the above inequalities for all $k \geq 0$, we have
\begin{align*}
f({\bf Z}_{0}) \geq \frac{1}{2} \zeta_{\mathrm{min}} \sum_{k = 0}^{\infty}\|{\bf Z}_{k} - {\bf Z}_{k+1}\|_{\mathrm{F}}^{2},
\end{align*}
which implies that $\lim_{k \to \infty}\|{\bf Z}_{k} - {\bf Z}_{k+1}\|_{\mathrm{F}} = 0$.}

\subsection{Proof of Theorem \ref{boundedNoiseless}}
\label{proof_boundedNoiseless}
For any matrices ${\bf Y}$ and ${\bf Z}$ $\in \mathbb{C}^{M \times T}$, we have
\begin{align}
\|[{\bf Y} \! - \! {\bf Z}, \mu{\bf 1}]\|_{2,1} = ~ & \sum_{i =1}^{M} \sqrt{\|({\bf Y} \! - \! {\bf Z})_{i,:}\|_{2}^{2} +\mu^{2}} \nonumber \\ 
\geq ~ & \sum_{i=1}^{M} |\mu| ~ = ~ |\mu| M \nonumber \\
\|[{\bf Z}, \mu{\bf I}]\|_{*} = ~ & \text{trace} \! \left( \left({\bf Z}{\bf Z}^{\text{H}} + \mu^{2}{\bf I}\right)^{\frac{1}{2}} \right) \nonumber \\
\label{ineq_nuclear}
\geq ~ & \left(\text{trace} \! \left({\bf Z}{\bf Z}^{\text{H}} + \mu^{2}{\bf I} \right) \right)^{\frac{1}{2}} \\
= ~ & \left(\text{trace} \! \left( {\bf Z}{\bf Z}^{\text{H}}\right) + M\mu^{2} \right)^{\frac{1}{2}} \nonumber \\
\geq ~ & \left( M\mu^{2} \right)^{\frac{1}{2}} ~ = ~ |\mu|\sqrt{M}. \nonumber
\end{align}
Inequality (\ref{ineq_nuclear}) holds because $\text{trace} \! \left( {\bf X}^{\frac{1}{2}} \right) \! = \! \sum_{i} \! \sqrt{\zeta_{i}} \geq \sqrt{\sum_{i} \zeta_{i} } = \left(\text{trace} \! \left( {\bf X} \right) \right)^{\frac{1}{2}}$ for any symmetric matrix ${\bf X}$, with $\zeta_{i}$ being the eigenvalue of ${\bf X}$. Therefore, the objective function in (\ref{prob_noiseless}) is bounded below as $f({\bf Z}) = \|[{\bf Z}, \mu{\bf I}]\|_{*} + \lambda \|[{\bf Y} \! - \! {\bf Z}, \mu{\bf 1}]\|_{2,1} \geq |\mu| (\sqrt{M} + \lambda M)$.

{\color{black}
\subsection{Proof of Theorem \ref{KKTpointNoiseless}}
\label{proof_KKTpointNoiseless}
Denote the limit point of sequence $\{{\bf Z}_{k}\}$ as ${\bf Z}_{k+1}$. Then, according to $\lim_{k \to \infty} \|{\bf Z}_{k} - {\bf Z}_{k+1}\|_{\mathrm{F}} = 0$ in Theorem \ref{nonincreasingNoiseless} and (\ref{irls_progress_noiseless}), we have $ {\bf Z}_{k+1} = \lambda({\bf P}_{k+1} + \lambda{\bf Q}_{k+1})^{-1}{\bf Q}_{k+1}{\bf Y}$, that is, ${\bf P}_{k+1}{\bf Z}_{k+1} + \lambda{\bf Q}_{k+1}({\bf Z}_{k+1} - {\bf Y}) = {\bf 0}$. This indicates that ${\bf Z}_{k+1}$ satisfies the KKT condition. Since Problem (\ref{prob_noiseless}) is convex w.r.t. ${\bf Z}$, the stationary point is globally optimal.
}

\subsection{Proof of Theorem \ref{nonincreasingNoisy}}
\label{proof_nonincreasingNoisy}
Similar to the proof of Theorem \ref{nonincreasingNoiseless}, we calculate the difference between the objective function values in two successive iterations as
\begingroup
\allowdisplaybreaks
\begin{align}
& f({\bf Z}_{k}, {\bf V}_{k}) - f({\bf Z}_{k+1}, {\bf V}_{k+1}) \nonumber \\
= ~ & \frac{1}{2}\|{\bf Y} \! - \! {\bf Z}_{k} \! - \! {\bf V}_{k}\|_{\text{F}}^{2} - \! \frac{1}{2}\|{\bf Y} \! - \! {\bf Z}_{k+1} \! - \! {\bf V}_{k+1}\|_{\text{F}}^{2} \nonumber \\
& + \lambda_{1}\|[{\bf Z}_{k}, \mu{\bf I}]\|_{*} - \lambda_{1}\|[{\bf Z}_{k+1}, \mu{\bf I}]\|_{*} \nonumber  \\
& + \lambda_{2}\|[{\bf V}_{k}, \mu{\bf 1}]\|_{2,1} - \lambda_{2}\|[{\bf V}_{k+1}, \mu{\bf 1}]\|_{2,1} \nonumber \\
\geq ~ & \frac{1}{2}\|{\bf Y} \! - \! {\bf Z}_{k} \! - \! {\bf V}_{k}\|_{\text{F}}^{2} - \! \frac{1}{2}\|{\bf Y} \! - \! {\bf Z}_{k+1} \! - \! {\bf V}_{k+1}\|_{\text{F}}^{2} \nonumber \\
& + \lambda_{1} ~\! \text{trace} \! \left( \frac{1}{2}\left({\bf Z}_{k}{\bf Z}_{k}^{\text{H}} - {\bf Z}_{k+1}{\bf Z}_{k+1}^{\text{H}}\right){\bf P}_{k} \right) \nonumber \\
& + \frac{\lambda_{2}}{2}\text{trace} \! \left( {\bf Q}_{k}\left({\bf V}_{k}{\bf V}_{k}^{\text{H}} - {\bf V}_{k+1}{\bf V}_{k+1}^{\text{H}} \right)  \right) \nonumber \\
= ~ & \frac{1}{2}\|{\bf Y} \! - \! {\bf Z}_{k} \! - \! {\bf V}_{k}\|_{\text{F}}^{2} - \! \frac{1}{2}\|{\bf Y} \! - \! {\bf Z}_{k+1} \! - \! {\bf V}_{k+1}\|_{\text{F}}^{2} \nonumber \\
& + \lambda_{1} ~\! \text{trace} \! \left( \frac{1}{2}\left({\bf Z}_{k} - {\bf Z}_{k+1}\right)\left({\bf Z}_{k} - {\bf Z}_{k+1}\right)^{\text{H}}{\bf P}_{k} \right)  \nonumber \\
& + \lambda_{1} ~\! \text{trace} \! \left( \left({\bf Z}_{k} - {\bf Z}_{k+1} \right){\bf Z}_{k+1}^{\text{H}} {\bf P}_{k} \right) \nonumber \\
& + \frac{\lambda_{2}}{2}\text{trace} \! \left( {\bf Q}_{k}\left({\bf V}_{k} - {\bf V}_{k+1} \right) \left({\bf V}_{k} - {\bf V}_{k+1} \right)^{\text{H}} \right) \nonumber \\
\label{difference_trace_noisy}
& + \lambda_{2} ~\! \text{trace} \! \left( {\bf Q}_{k}({\bf V}_{k} - {\bf V}_{k+1}){\bf V}_{k+1}^{\text{H}} \right) \\
\geq ~ & \frac{1}{2}\text{trace} \! \left( ({\bf Z}_{k} - {\bf Z}_{k+1})({\bf Z}_{k} - {\bf Z}_{k+1})^{\text{H}} \right. \nonumber \\
& \qquad \quad \left. + ~ ({\bf V}_{k} - {\bf V}_{k+1})({\bf V}_{k} - {\bf V}_{k+1})^{\text{H}} \right) \nonumber \\
& + \text{trace} \! \left( (-{\bf Y} + {\bf Z}_{k+1})({\bf Z}_{k} - {\bf Z}_{k+1})^{\text{H}}\right. \nonumber \\
& \qquad \quad \left. + ~ (-{\bf Y} + {\bf V}_{k+1})({\bf V}_{k} - {\bf V}_{k+1})^{\text{H}} \right. \nonumber \\
& \qquad \quad \left. + ~ {\bf Z}_{k}{\bf V}_{k}^{\text{H}} - {\bf Z}_{k+1}{\bf V}_{k+1}^{\text{H}} \right) \nonumber \\
& + \lambda_{1} ~\! \text{trace} \! \left( \left({\bf Z}_{k} - {\bf Z}_{k+1} \right){\bf Z}_{k+1}^{\text{H}} {\bf P}_{k} \right) \nonumber \\ 
& + \lambda_{2} ~\! \text{trace} \! \left( {\bf Q}_{k}({\bf V}_{k} - {\bf V}_{k+1}){\bf V}_{k+1}^{\text{H}} \right) \nonumber \\
= ~ & \frac{1}{2}\text{trace} \! \left( ({\bf Z}_{k} - {\bf Z}_{k+1})({\bf Z}_{k} - {\bf Z}_{k+1})^{\text{H}} \right. \nonumber \\
& \qquad \quad \left. + ~ ({\bf V}_{k} - {\bf V}_{k+1})({\bf V}_{k} - {\bf V}_{k+1})^{\text{H}} \right) \nonumber \\
& + \text{trace} \! \left( {\bf Z}_{k}{\bf V}_{k}^{\text{H}} - {\bf Z}_{k+1}{\bf V}_{k+1}^{\text{H}} \right) \nonumber \\
& + \text{trace} \! \left( -{\bf V}_{k}({\bf Z}_{k} - {\bf Z}_{k+1})^{\text{H}} \right) \nonumber \\
& + \text{trace} \! \left( -{\bf Z}_{k}({\bf V}_{k} - {\bf V}_{k+1})^{\text{H}} \right) \nonumber \\
= ~ & \frac{1}{2}\text{trace} \! \left( ({\bf Z}_{k} - {\bf Z}_{k+1})({\bf Z}_{k} - {\bf Z}_{k+1})^{\text{H}} \right. \nonumber \\
& \qquad \quad \left. + ~ ({\bf V}_{k} - {\bf V}_{k+1})({\bf V}_{k} - {\bf V}_{k+1})^{\text{H}} \right) \nonumber \\
& + \text{trace} \! \left( ({\bf Z}_{k} - {\bf Z}_{k+1})({\bf V}_{k+1} - {\bf V}_{k})^{\text{H}} \right) \nonumber \\
= ~ & \frac{1}{2}\text{trace} \! \left( ({\bf Z}_{k} - {\bf Z}_{k+1} - {\bf V}_{k} + {\bf V}_{k+1}) \cdot \right. \nonumber \\
& \qquad \quad \left. ({\bf Z}_{k} - {\bf Z}_{k+1} - {\bf V}_{k} + {\bf V}_{k+1})^{\text{H}} \right) \nonumber \\
= ~ & \frac{1}{2}\| {\bf Z}_{k} - {\bf Z}_{k+1} - {\bf V}_{k} + {\bf V}_{k+1} \|_{\text{F}}^{2} ~ \geq ~ 0, \nonumber
\end{align}
\endgroup
which indicates that $f({\bf Z} , {\bf V})$ is a non-increasing function. {\color{black}Since $f({\bf Z} , {\bf V})$ is non-increasing, we have
\begin{align*}
& \min\{\lambda_{1}, \lambda_{2}\}\left(\|[{\bf Z}_{k}, \mu{\bf I}]\|_{*}+ \|[{\bf V}_{k}, \mu{\bf 1}]\|_{2,1} \right) \\
\leq ~ & \lambda_{1}\|[{\bf Z}_{k}, \mu{\bf I}]\|_{*}+ \lambda_{2}\|[{\bf V}_{k}, \mu{\bf 1}]\|_{2,1} \\
\leq ~ & f({\bf Z}_{k}, {\bf V}_{k}) ~ \leq ~ f({\bf Z}_{0}, {\bf V}_{0}).
\end{align*}
Hence, $\|{\bf Z}_{k}\|_{*} + \|{\bf V}_{k}\|_{2,1} < \|[{\bf Z}_{k}, \mu{\bf I}]\|_{*} + \|[{\bf V}_{k}, \mu{\bf 1}]\|_{2,1} \leq \frac{f({\bf Z}_{0}, {\bf V}_{0})}{\min\{\lambda_{1}, \lambda_{2}\}}$, which indicates that the sequence $\{({\bf Z}_{k} , {\bf V}_{k})\}$ is bounded.

Besides, combining (\ref{irls_progress_noisy}) and (\ref{difference_trace_noisy}) yields
\begingroup
\allowdisplaybreaks
\begin{align*}
& f({\bf Z}_{k}, {\bf V}_{k}) - f({\bf Z}_{k+1}, {\bf V}_{k+1}) \\
\geq ~ & \frac{\lambda_{1}}{2} ~\! \text{trace} \! \left( \left({\bf Z}_{k} - {\bf Z}_{k+1}\right)\left({\bf Z}_{k} - {\bf Z}_{k+1}\right)^{\text{H}}{\bf P}_{k} \right)   \\
& + \frac{\lambda_{2}}{2}\text{trace} \! \left( {\bf Q}_{k}\left({\bf V}_{k} - {\bf V}_{k+1} \right) \left({\bf V}_{k} - {\bf V}_{k+1} \right)^{\text{H}} \right)  \\
& +  \frac{1}{2}\text{trace} \! \left( ({\bf Z}_{k} - {\bf Z}_{k+1})({\bf Z}_{k} - {\bf Z}_{k+1})^{\mathrm{H}} \right) \\
\geq ~ & \frac{\lambda_{1}}{2} ~\! \text{trace} \! \left( \left({\bf Z}_{k} - {\bf Z}_{k+1}\right)\left({\bf Z}_{k} - {\bf Z}_{k+1}\right)^{\text{H}}{\bf P}_{k} \right)   \\
& + \frac{\lambda_{2}}{2}\text{trace} \! \left( {\bf Q}_{k}\left({\bf V}_{k} - {\bf V}_{k+1} \right) \left({\bf V}_{k} - {\bf V}_{k+1} \right)^{\text{H}} \right) \\
\geq ~ & \frac{\lambda_{1}}{2}\sum_{i}^{M}\zeta_{i}({\bf P}_{k}) \zeta_{M - i + 1}( ({\bf Z}_{k} \!-\! {\bf Z}_{k+1})({\bf Z}_{k} \!-\! {\bf Z}_{k+1})^{\mathrm{H}} ) \\
& + \frac{\lambda_{2}}{2}\sum_{i}^{M} \! \zeta_{i}({\bf Q}_{k}) \zeta_{M \! - i + 1} \! ( \! ( \! {\bf V}_{k} \!-\! {\bf V}_{k+1} \! )( \! {\bf V}_{k} \!-\! {\bf V}_{k+1} \!)^{\mathrm{H}} \! ) \\
\geq ~ & \frac{\lambda_{1}}{2} \zeta_{\mathrm{min}}^{({\bf P})} \times \|{\bf Z}_{k} \!-\! {\bf Z}_{k+1}\|_{\mathrm{F}}^{2} + \frac{\lambda_{1}}{2} \zeta_{\mathrm{min}}^{({\bf Q})} \times \|{\bf V}_{k} \!-\! {\bf V}_{k+1}\|_{\mathrm{F}}^{2},
\end{align*}
where $ \zeta_{\mathrm{min}}^{({\bf P})}$ and $ \zeta_{\mathrm{min}}^{({\bf Q})}$ are the smallest eigenvalues of ${\bf P}_{k}$ and ${\bf Q}_{k}$, respectively, over all $k$. Summing all the above inequalities for all $k \geq 0$, we have
\begin{align*}
f({\bf Z}_{0}, {\bf V}_{0}) \geq & ~ \frac{\lambda_{1}}{2} \zeta_{\mathrm{min}}^{({\bf P})} \sum_{k = 0}^{\infty} \|{\bf Z}_{k} - {\bf Z}_{k+1}\|_{\mathrm{F}}^{2} \\
& +  \frac{\lambda_{2}}{2} \zeta_{\mathrm{min}}^{({\bf Q})} \sum_{k = 0}^{\infty} \|{\bf V}_{k} - {\bf V}_{k+1}\|_{\mathrm{F}}^{2},
\end{align*}
\endgroup
which implies that $\lim_{k \to \infty} \|{\bf Z}_{k} - {\bf Z}_{k+1}\|_{\mathrm{F}} = 0$ and $\lim_{k \to \infty} \|{\bf V}_{k} - {\bf V}_{k+1}\|_{\mathrm{F}} = 0$.}

\subsection{Proof of Theorem \ref{boundedNoisy}}
\label{proof_boundedNoisy}
Considering $\|{\bf Y} - {\bf Z} - {\bf V}\|_{\text{F}}^{2} \geq 0$ and the inequalities in Appendix \ref{proof_boundedNoiseless}, we can prove that the objective function in (\ref{prob_f_ZV}) is bounded below as $f({\bf Z} , {\bf V}) = \frac{1}{2}\|{\bf Y} \! - \! {\bf Z} \! - \! {\bf V}\|_{\text{F}}^{2} + \lambda_{1}\|[{\bf Z}, \mu{\bf I}]\|_{*} + \lambda_{2}\|[{\bf V}, \mu{\bf 1}]\|_{2,1} \geq |\mu|(\lambda_{1}\sqrt{M} + \lambda_{2}M)$.

{\color{black}
\subsection{Proof of Theorem \ref{KKTpointNoisy}}
\label{proof_KKTpointNoisy}
Denote the limit point of the sequence $\{({\bf Z}_{k}, {\bf V}_{k})\}$ as $({\bf Z}_{k+1}, {\bf V}_{k+1})$. Then, according to $\lim_{k \to \infty} \|{\bf Z}_{k} - {\bf Z}_{k+1}\|_{\mathrm{F}} = 0$ and $\lim_{k \to \infty} \|{\bf V}_{k} - {\bf V}_{k+1}\|_{\mathrm{F}} = 0$ in Theorem \ref{nonincreasingNoisy} and (\ref{irls_progress_noisy}), we have
\begin{align*}
\left\{ \!\!
\begin{array}{l}
{\bf Z}_{k+1} = ({\bf I} + \lambda_{1}{\bf P}_{k+1})^{-1}({\bf Y} - {\bf V}_{k+1}) \\
{\bf V}_{k+1} \! = ({\bf I} + \lambda_{2}{\bf Q}_{k+1})^{-1}({\bf Y} - {\bf Z}_{k+1}),
\end{array}
\right.
\end{align*}
which is the KKT condition of Problem (\ref{prob_f_ZV}). Since Problem (\ref{prob_f_ZV}) is convex w.r.t. ${\bf Z}$ and ${\bf V}$, the stationary point is globally optimal.
}

\section*{ACKNOWLEDGMENT}

The work of Huiping Huang is supported by the Graduate School CE within the Centre for Computational Engineering at Technische Universität Darmstadt.

{\footnotesize{
\bibliographystyle{IEEEtran}
\balance
\bibliography{strings, refs}
}}

\end{document}